\documentclass[a4paper,twocolumn,10pt,accepted=2023-09-11]{quantumarticle}
\pdfoutput=1

\usepackage[numbers,compress]{natbib}
\usepackage{hyperref}

\usepackage[T1]{fontenc}
\usepackage[normalem]{ulem}
\usepackage{exscale}
\usepackage{graphicx}
\usepackage{amsmath}
\usepackage{latexsym}
\usepackage{epstopdf}
\usepackage{amsfonts}
\usepackage{amssymb}
\usepackage{bbm}
\usepackage{lipsum}
\usepackage{amsthm}
\usepackage{fancyhdr,txfonts,bbm}
\usepackage{appendix}
\usepackage{soul}
\usepackage{float}
\usepackage{changes}
\usepackage{epsfig,pstricks}

\newtheorem{lemma}{Lemma}
\newtheorem{corollary}{Corollary}

\newtheorem{proposition}{Proposition}

\newcommand{{\Cd}}{{\mathbb{C}^d}}
\newcommand{{\C}}{{\mathbb{C}}}

\newcommand{\ketbra}[2]{| #1\rangle\!\langle #2|}
\newcommand{\tr}[1]{\mbox{Tr}\left[ #1\right]}
\newcommand{\trr}[1]{\mbox{Tr}[ #1 ]}

\newcommand{\expp}[1]{e^{#1}}

\newcommand{\mL}{{\mathbf{\Lambda}}}
\newcommand{\omL}{\overline{\mathbf{\Lambda}}}
\newcommand{\TL}{{T^\Lambda}}

\begin{document}

\title{Pure non-Markovian evolutions}

\author{Dario De Santis}
\affiliation{ICFO-Institut de Ciencies Fotoniques, The Barcelona Institute of Science and Technology, 08860 Castelldefels (Barcelona), Spain}
\affiliation{Scuola Normale Superiore, I-56126 Pisa, Italy}
\orcid{0000-0001-5666-1770}
%\date{\today}
\maketitle

\begin{abstract}
{\bf
Non-Markovian dynamics are characterized by information backflows, where the evolving open quantum system retrieves part of the information previously lost in the environment. Hence, the very definition of non-Markovianity implies an initial time interval when the evolution is noisy, otherwise no backflow could take place. 
We identify two types of initial noise, where the first has the only effect of degrading the information content of the system, while the latter is essential for the appearance of non-Markovian phenomena. 
Therefore, all non-Markovian evolutions can be divided into two classes: noisy non-Markovian (NNM), showing both types of noise, and pure non-Markovian (PNM), implementing solely essential noise. 
We make this distinction through a timing analysis of fundamental non-Markovian features. 
First, we prove that all NNM dynamics can be simulated through a Markovian pre-processing of a PNM core. We quantify the gains in terms of information backflows and non-Markovianity measures provided by PNM evolutions. 
Similarly, we study how the entanglement breaking property behaves in this framework and we discuss a technique to activate correlation backflows. Finally, we show the applicability of our results through the study of several well-know dynamical models. 
}
\end{abstract}

\section{Introduction} 
Open quantum system dynamics describe the evolution of quantum systems interacting with an external system, typically represented by the surrounding environment. The unavoidable nature of this interaction made this topic of central interest in the field of quantum information~\cite{book_B&P,book_R&H}.
This reciprocal action may lead to two different regimes for the information initially stored into our system. 
An evolution is called Markovian whenever there are no memory revivals and therefore the system shows a monotonic information degradation.
On the contrary, non-Markovian evolutions are those showing information backflows, where partial information stored into the system is first lost into the environment and then retrieved at later times (for reviews on this topic see~\cite{REVDARIUSZ,rev_RHP,rev_BLP,revmod2,rev4}).
Hence, the very definition of these evolutions implies the existence of an initial time interval when the dynamics is noisy, otherwise no backflow from the environment could be possible.

In this work, we address the question of whether all the initial noise applied by an evolution is necessary for the following non-Markovian phenomena. 
We identify two noise types. While the first, that we call \textit{useless}, is not necessary for information backflows, only the information lost with \textit{essential} noise takes part to the characteristic non-Markovian phenomena. 
Starting from this observation, we classify non-Markovian evolutions as \textit{noisy} or \textit{pure}, where the first have both types of noise, while the second implements essential noise only. Hence, the information initially lost with pure non-Markovian (PNM) evolutions always take part to a later backflow, which occur even in time intervals starting immediately after the beginning of the interaction with the environment. Instead, the useless noise of noisy non-Markovian (NNM) evolutions has the sole result of damping the information content of the open system and it diminishes the amplitude of backflows.

This classification is in close analogy with the structure of quantum states, where mixed states can be obtained through noisy operations on pure states. 
Similarly, NNM evolutions can be obtained via Markovian pre-processings  of  PNM evolutions, which we call PNM \textit{cores}. 
Moreover, as well as pure states allow the best performances in several scenarios and protocols, PNM evolutions are characterized by the largest information revivals and non-Markovianity measures.

The interest in considering PNM cores of known NNM evolutions resides in the possibility to isolate a dynamics with the same non-Markovian qualitative features and at the same time with the largest possible non-Markovian phenomena. For instance, in case of an experimental setup where the visible non-Markovian phenomena generated by a target evolution are not significant due to various additional noise sources in the laboratory (preparation, measurements, thermal noise, ecc...), the possibility to isolate and implement the corresponding PNM core may allow to appreciate the same non-Markovian phenomena that we failed to detect with the noisy version.

The first main goal of this work is to identify the initial useless noise of generic non-Markovian evolutions. While doing so, we propose a structure for the timing of the fundamental non-Markovian phenomena happening in finite and infinitesimal time intervals.
This framework provides a straightforward and natural approach to discriminate Markovian, NNM and PNM evolutions.
We follow by showing how to isolate the PNM core of a generic NNM evolution and, conversely, how any PNM evolution can generate a whole class of NNM evolutions. 
Later, we explain how and to what extent PNM evolutions are characterized by larger information backflows and non-Markovianity measures. In particular, we focus on backflows of state distinguishability. Finally, we show how the entanglement breaking property behaves in this scenario and we discuss a technique to activate correlation backflows that cannot be observed in presence of useless noise. Later we apply our results to several models, such as depolarizing and dephasing evolutions.

\section{Quantum evolutions}\label{evolutions}

We define $S(\mathcal{H})$ to be the set of density matrices of a generic $d$-dimensional Hilbert space $\mathcal{H}$. The time evolution of any open quantum system can be represented by a one-parameter family $\mL=\{\Lambda_t\}_{t\geq 0}$ of quantum maps, namely completely positive and trace preserving (CPTP) superoperators. 
We define $\mL$ to be the \textit{evolution} of the system, while $\Lambda_t$ is the corresponding \textit{dynamical map} at time $t$.
Hence, the transformation of an initial state $\rho(0)\in S(\mathcal{H})$ into the corresponding evolved state at time $t$ is $ \rho(t)=\Lambda_t(\rho(0))\in S(\mathcal{H})$.

We consider $\mL$ as a collection of dynamical maps \textit{continuous} in time. This is because any open quantum system evolution obtained through the physical interaction with an environment, even in case of non-continuous Hamiltonians, are continuous~\cite{continuity}. Secondly, we assume \textit{divisibility}, namely the existence of an \textit{intermediate map} for any time interval. More precisely, for all $0\leq s \leq t$ we assume the existence of a linear map $V_{t,s}$ such that $\Lambda_t=V_{t,s}\circ \Lambda_s$. Invertible evolutions are an instance of divisible evolutions. We call an evolution invertible if, for all $t\geq 0$, there exists the operator $\Lambda_t^{-1}$ such that $\Lambda_t^{-1}\circ\Lambda_t=I$, where $I$ is the identity map on $S(\mathcal{H})$.  Indeed, in these cases $V_{t,s}=\Lambda_t\circ \Lambda_s^{-1}$. While divisibility makes all the steps of the following sections easier, in Section~\ref{nondivisible} we show how to generalize our results to non-divisible evolutions.

We say that an evolution is \textit{CP-divisible} if and only if between any two times it is represented by a quantum channel. 
Hence, this property corresponds to require that the intermediate maps $V_{t,s}$ are CPTP for all $0\leq s\leq t$. Remember that any implementable quantum operation is represented by a CPTP operator, which is the reason why dynamical maps $\Lambda_t$ are required to be CPTP at all times. In case of non-CP-divisible evolutions, there exist $s\leq t$ such that $V_{t,s}$ is not CPTP, while at the same time $\Lambda_t=V_{t,s}\circ \Lambda_s$ must be CPTP. In this case the transformation acting in the time interval $[s,t]$ cannot be applied independently from the transformation applied in $[0,s]$. 

We define Markovian evolutions as those being CP-divisible. 
Thanks to the Stinespring-Kraus representation theorem~\cite{Stine,Kraus}, such a definition adheres with the impossibility of the system to recover any information that was previously lost. Indeed, as we better explain later, CPTP operators degrades the information content of the system.
Therefore, an evolution is non-Markovian if and only if there exists at least one time interval $[s,t]$ when the evolution is not described by a CPTP intermediate map. Indeed, whenever this is the case, it is possible to obtain an information backflow during the same time interval~\cite{bogna,DDSshort}, even for non-invertible evolutions~\cite{BD,DDSMJ}.
 
Given an evolution $\mL$, we define $\mathcal P^\Lambda$ to be the collection of time pairs such that the corresponding intermediate maps are CPTP. 
This set can be obtained by considering the smallest eigenvalue $\lambda_{t,s}$ of the state obtained by applying the  Choi-Jamiołkowski isomorphism  to $V_{t,s}$~\cite{Choi,Jamilokowski}:
\begin{eqnarray}\label{CPTPsub}
\!\!\!\!\!\!\!\!\!\!\! \mathcal{P}^\Lambda \!\!\!&=&\!\!\! \{\{s,t\} \,|\,\, 0\leq s\leq t \,\mbox{ and }\, V_{t,s}\, \mbox{ is CPTP} \} \\ \nonumber
\!\!\!&=&\!\!\! \{\{s,t\} \,|\,\, 0\leq s\leq t \,\mbox{ and }\, \lambda_{t,s}\geq 0 \}.
\end{eqnarray}
Indeed, $V_{t,s}$ is CPTP if and only if $\lambda_{t,s}\geq 0$. 
Similarly, we define the complementary set of $\mathcal P^\Lambda$ to be $\mathcal{N}^\Lambda$: the collection of time pairs such that $V_{t,s}$ is non-CPTP:
\begin{eqnarray}
\label{nonCPTPsub}
\!\!\!\!\!\mathcal{N}^\Lambda \!\!\!&=&\!\!\! \{\{s,t\} \,|\,\, 0\leq s\leq t \,\mbox{ and }\,V_{t,s} \mbox{ is not CPTP} \} \\
 \nonumber  \!\!\!&=&\!\!\! \{\{s,t\} \,|\,\, 0\leq s\leq t \,\mbox{ and }\, \lambda_{t,s}< 0 \}  .
\end{eqnarray}

An evolution $\mL$ is Markovian if and only if $\mathcal{N}^\Lambda$ is empty. The border of $\{s,t\}_{0\leq s\leq t} $ always belongs to $\mathcal{P}^\Lambda$: the vertical line $\{0,t\}_{t\geq 0}$ corresponds to the (CPTP) dynamical maps and the diagonal line $\{t,t\}_{t\geq 0}$ corresponds to the trivial intermediate (identity) maps.  The pairs infinitesimally close to $\{t,t\}_{t\geq 0}$ correspond to the infinitesimal intermediate maps $V_{t+\epsilon,t}$, where their CPTP/non-CPTP nature can be studied through the corresponding master equation rates~\cite{RHP,darekk} (see Section~\ref{eternality}).
Instead, any point in the interior of  $\{s,t\}_{0\leq s\leq t}$ can either belong to $\mathcal P^\Lambda$ or $\mathcal N^\Lambda$ but not every open set is allowed for $\mathcal N^\Lambda$: some constraints deriving from fundamental map composition rules have to be satisfied\footnote{Consider $V_{t_3,t_1}=V_{t_3,t_2}\circ V_{t_2,t_1}$ for $t_1< t_2< t_3$: 
if $V_{t_2,t_1}$ and $ V_{t_3,t_2}$ are CPTP, then  $V_{t_3,t_1}$ is CPTP; 
if $V_{t_3,t_1}$ is non-CPTP and $V_{t_2,t_1}$ is CPTP, then $V_{t_3,t_2}$ is non-CPTP;  if $V_{t_3,t_1}$ is non-CPTP and $V_{t_3,t_2}$ is CPTP, then $V_{t_2,t_1}$ is non-CPTP; if $V_{t_2,t_1}$ is non CPTP, there exists $t\in (t_1,t_2)$ such that $V_{t+\epsilon,t}$ is non-CPTP for infinitesimal $\epsilon>0$.}. Finally, we prove that $\mathcal{P}^\Lambda$ is closed and $\mathcal{N}^\Lambda$ is open in Appendix~\ref{proofclosed}. Below, we show several representations of $\mathcal P^\Lambda$ and $\mathcal N^\Lambda$.

\section{Timing of information backflows}

In this section we show how the timing of the main non-Markovian phenomena of an evolution are always ruled by three times: $\TL$, $\tau^\Lambda$ and $t^\Lambda$. We start by explaining their operational meaning.

It is possible to observe an information backflows during a time interval if and only if it starts later than $\TL$. Hence, there exist intervals $[\TL+\epsilon, t]$ when the corresponding intermediate maps are not CPTP for infinitesimal $\epsilon>0$\footnote{When we say that a feature is satisfied for an ``infinitesimal $\epsilon$'', we mean that for a given $\epsilon^*>0$ this feature holds true for all $\epsilon \in (0,\epsilon^*)$.}. Among these time intervals, $[\TL+\epsilon, t^\Lambda]$ is the shortest, namely $t=t^\Lambda$ is the earliest final time such that $V_{t^\Lambda, \TL+\epsilon}$ is not CPTP for infinitesimal $\epsilon>0$. Instead, $\tau^\Lambda$ is the earliest time when an instantaneous backflow can be observed, namely the earliest $t$ such that $V_{t+\epsilon,t}$ is not CPTP for infinitesimal $\epsilon>0$. Hence,  $[\TL+\epsilon, t^\Lambda]$ ($[\tau^\Lambda,\tau^\Lambda+\epsilon]$) is the shortest time interval with the earliest \textit{initial} (\textit{final}) time when the corresponding intermediate map is not CPTP. 
For these reasons, we call the noise applied by the evolution in $[0,\TL]$ as {useless} for non-Markovianity, while the noise applied later than $\TL$ is {essential} for non-Markovian phenomena.
We represent the typical role of these three times in Fig. \ref{disegnetto}.

\begin{figure}
\includegraphics[width=0.45\textwidth]{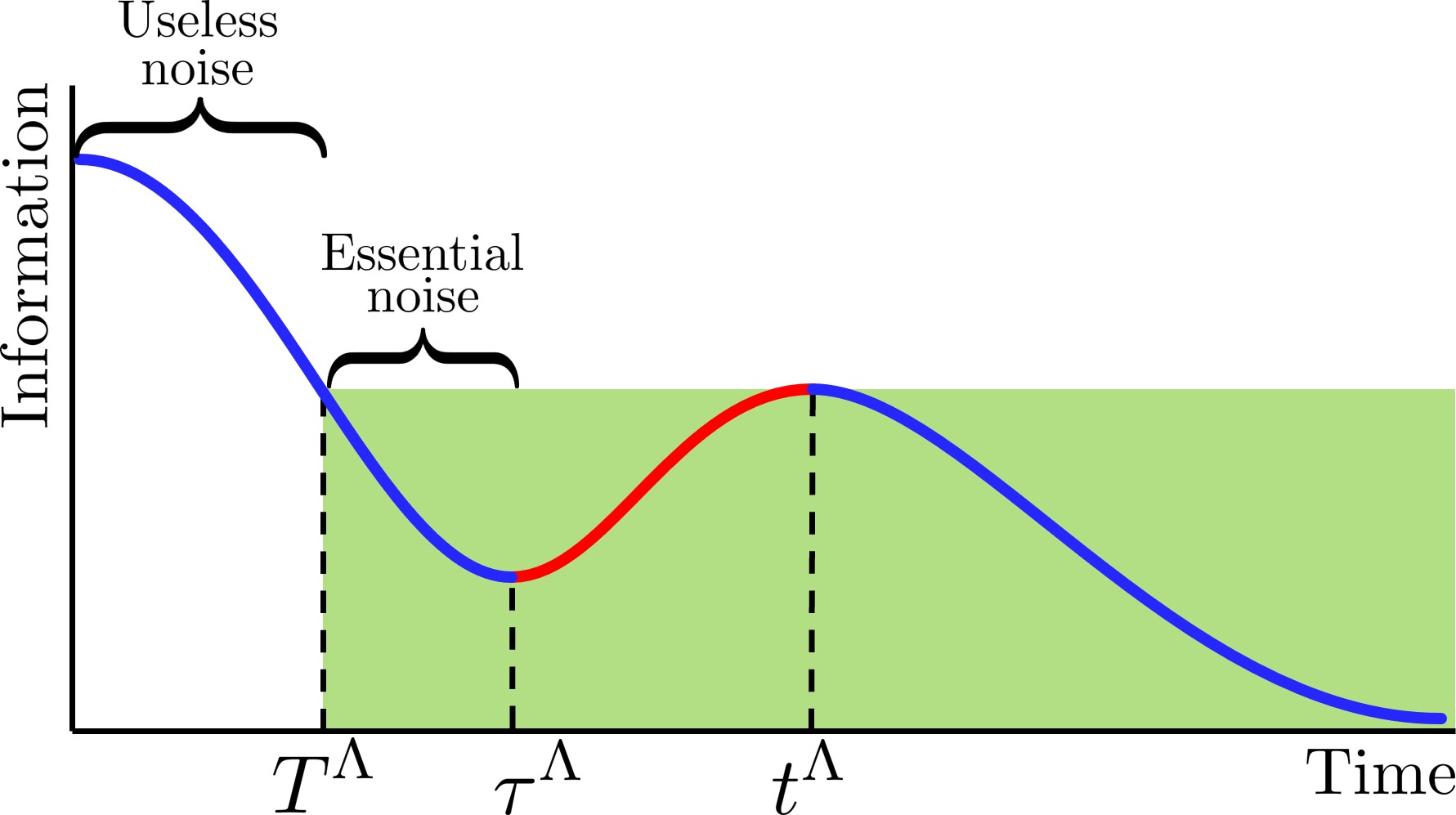}
\caption{Typical information content of an open quantum system evolving under a NNM evolution. An increase, or backflow, of information, is a typical sign of non-Markovianity, namely of non-CPTP intermediate maps. Blue/red regions represent times when the infinitesimal intermediate map $V_{t+\epsilon,t}$ is CPTP/non-CPTP. The time $T^\Lambda$ is the largest such that the preceding dynamics is CP-divisible and $V_{t,T^\Lambda}$ is CPTP for all $t\geq T^\Lambda$. Indeed, for $t\geq \TL$, the information content of the system never exceeds the level at $\TL$ (green area). The information lost in $[0,T^\Lambda]$ is never recovered (useless noise), while the noise applied in $[T^\Lambda,\tau^\Lambda]$ is essential for the following backflows. $\tau^\Lambda$ is the earliest time after which we have an instantaneous backflow. We have finite backflows in intervals $[s,t]$ with $s>\TL$ and $t^\Lambda$ is the earliest $t$ such that we have a  backflow in $[\TL+\epsilon,t]$.
\label{disegnetto} }
\end{figure}

We follow by giving the mathematical definitions of $T^\Lambda$, $\tau^\Lambda$ and $t^\Lambda$. Given a generic evolution $\mL$, we define:
\begin{equation}\label{TLambdanuovo}
\!\!\! T^{\Lambda}= \max \left\{\, T\, \left|   
\begin{array}{ccc}
\!\mbox{(A)} &\!\! V_{t,s} \,\,\mbox{ CPTP for all} &\!\!\!  s\leq t \leq T\\ 
\!\mbox{(B)} &\!\! V_{t,T} \,\,\mbox{ CPTP for all} &\!\!\! T\leq t  \\ 
\!\mbox{(C)} &\!\! \Lambda_T \,\, \mbox{ not unitary} &\!\!\! { T>0}
\end{array}
\right. \!\!\! \right\}  .
\end{equation}
We briefly discuss conditions~(A), (B) and (C). Condition (A) requires the evolution to be CP-divisible before $\TL$: \textit{no non-Markovian effects can take place in $[0,\TL]$}. Condition (B) requires the evolution following $T^\Lambda$ to be physical \textit{by its own}, namely such that the composition with the initial noise $\Lambda_{T^\Lambda}$ is not needed for the intermediate maps $\{V_{t,T^\Lambda}\}_{t\geq T^\Lambda}$ to be CPTP. Finally, condition~(C) is imposed because a unitary transformation is not detrimental for the information content of our system and we cannot consider it useless noise: it is ``useless'' (for non-Markovian phenomena) but not noisy. We remember that evolutions with dynamical maps that are unitary at all times can be simulated with closed quantum systems, and therefore we do not focus on these cases or where condition (C) is necessary.
 In Appendix \ref{proofclosed} we show that Eq.~(\ref{TLambdanuovo}) is indeed a maximum and not a supremum. 

An evolution $\mL$ is Markovian if and only if $T^\Lambda=+\infty$. Indeed, all the noise applied by the evolution is not necessary for the later evolution to be physical. Markovian evolutions can be interpreted as sequences of noisy independent operations. Indeed, the dynamics between any two times the evolution is represented by a (noisy) CPTP map.
Below we show that a finite value of $\TL\geq 0$ implies the evolution to have non-CPTP intermediate maps and therefore to be non-Markovian. From now on, by $T^\Lambda\geq 0$ we mean $T^\Lambda\in [0,\infty)$. Hence, the time $\TL$ can be used to classify  quantum evolutions as follows:
\begin{itemize}
\item   Markovian: $T^\Lambda=+\infty$;
\item  Noisy non-Markovian (NNM): $T^\Lambda \in (0,+\infty)$; 
\item  Pure non-Markovian (PNM): $T^\Lambda=0$. 
\end{itemize}
Finally, the following results clarify the role of 
  $\TL$ (proofs in Appendix \ref{lemma123}):
\begin{lemma}\label{ABviolation}
Conditions (A), (B) and (C) are simultaneously satisfied at time $T$ if and only if $T\in [0,\TL]$. If (A) is violated at time $T$, (B) is violated at a strictly earlier time. 
\end{lemma}
\begin{lemma}\label{lemma4}
Any non-Markovian evolution $\mL$ has non-CPTP intermediate maps $V_{T,\TL+\epsilon}$ for one or more final times $T>\TL$ and infinitesimal values of $\epsilon >0$. If $V_{t,s}$ is non-CPTP, then $\TL<s$.
\end{lemma}

We follow by defining $\tau^\Lambda$ as the time when information begins to be instantaneously retrieved from the environment. Hence, it is defined by the earliest time when  $V_{T+\epsilon,T}$ is non-CPTP for infinitesimal $\epsilon>0$:
\begin{equation}\label{taunm}
\tau^{\Lambda}=\lim_{\epsilon\rightarrow 0^+} {\inf} \left\{ T \,|  \, V_{T+\epsilon,T} \mbox{ not CPTP}  \right\} .
\end{equation}

The time intervals with the earliest initial time such that the corresponding intermediate maps are non-CPTP are of the form $[\TL+\epsilon,t]$ (see Lemma~\ref{lemma4}). We define $t^{\Lambda}$ to be the earliest final time $t$ such that the corresponding intermediate map is non-CPTP:
\begin{equation}\label{tnm}
t^{\Lambda}=\lim_{\delta\rightarrow 0^+} \inf \left\{ T \,|  \, V_{T,\TL+\delta} \mbox{ not CPTP}  \right\}.
\end{equation}

The timing of the earliest information backflows is therefore dictated by $\TL$, $\tau^\Lambda$ and $t^\Lambda$, which have definite values for all non-Markovian evolutions.
 These three characteristic times satisfy the following reciprocal relation  (proof in Appendix \ref{tlambdavari}):
\begin{equation}\label{order}
0\leq \TL \leq \tau^\Lambda \leq t^\Lambda \leq \infty  .
\end{equation} 
We briefly discuss the possible equalities that can hold in the above equation.
$\TL=0$ corresponds to PNM evolutions, which are largely analysed throughout this work. A divergent $t^\Lambda=\infty$ is allowed only if $ \TL < \tau^\Lambda < t^\Lambda$.  For what concerns the possible equalities
between $\TL$, $\tau^\Lambda$ and $ t^\Lambda$, we have that $\TL= \tau^\Lambda$ implies $\TL= \tau^\Lambda = t^\Lambda$ (proof in Appendix \ref{tlambdavari}), namely 
$\TL = \tau^\Lambda < t^\Lambda$ is forbidden. 
Instead, $\TL< \tau^\Lambda = t^\Lambda$ is allowed.

We mention some examples for the above-mentioned patterns of $\TL$, $\tau^\Lambda$ and $t^\Lambda$. Concerning the difference between $\TL=0$ and $\TL>0$, we show how to obtain PNM evolutions ($\TL=0$) from NNM evolutions ($\TL>0$) and \begin{widetext}

\begin{figure}
\,\,\,\, \includegraphics[width=0.95\textwidth]{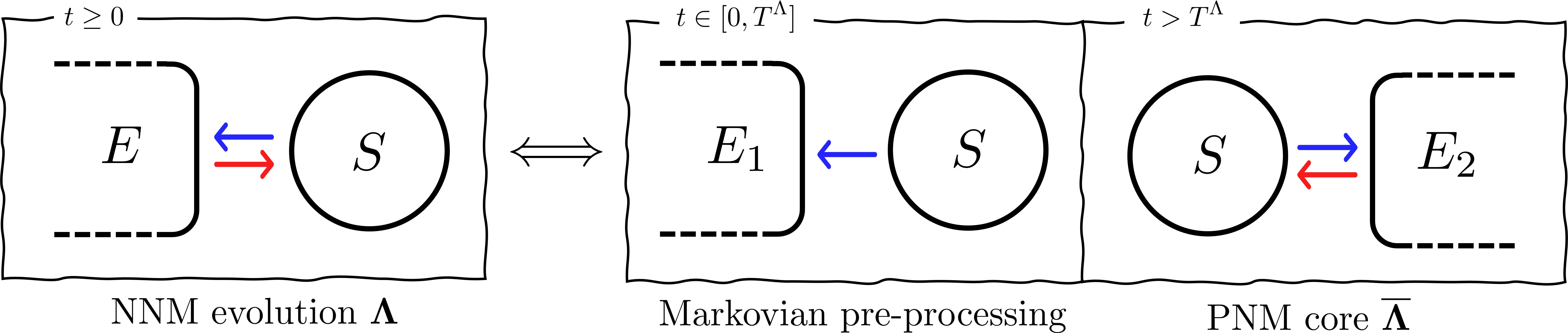}
\caption{Open quantum systems are physically represented by a system $S$ interacting with a surrounding environment $E$. This interaction leads to information losses (blue arrows) and, in case of non-Markovian evolutions, backflows (red arrows).
NNM evolutions $\mL$ can be simulated with the following two-stage scenario. First stage ($t\in[0,\TL]$): the system interacts with a first environment $E_1$ and information is lost monotonically (Markovian pre-processing).  The dynamics during this first stage corresponds to the useless noise of $\mL$.  
Second stage ($t>\TL$): $E_1$ is discarded, the system evolves while interacting with $E_2$ and we have information backflows. The dynamics during this second stage corresponds to the PNM core $\omL$ of $\mL$.  
}\label{resume}
\end{figure}
\end{widetext}
vice-versa in Section~\ref{secPNM}. 
We propose and in-depth study of $\TL < \tau^\Lambda < t^\Lambda < \infty$ for depolarizing evolutions in Section~\ref{secdepo}.
A well-known instance of $0=\TL= \tau^\Lambda = t^\Lambda $ is given by the eternal non-Markovian model~\cite{eternal0,eternal,eternal1}, while $\TL< \tau^\Lambda < t^\Lambda=\infty$ and $\TL< \tau^\Lambda = t^\Lambda<\infty$ are obtained with quasi-eternal non-Markovian evolutions~\cite{DDSlong}  in Section~\ref{eternality}.

The following result generalizes Lemma~\ref{lemma4}:
\begin{proposition}\label{propDgen}
Let $\mL$ be a generic non-Markovian evolution. If  $\TL<t^\Lambda$, then $  V_{t^\Lambda,s}$ is non-CPTP for all $s\in(\TL,t^\Lambda)$. If $\TL=t^\Lambda$, for all $T>\TL$ the infinitesimal intermediate map $V_{t+\epsilon,t}$ is non-CPTP for an infinite number of times $t\in(\TL,T)$.
\end{proposition} 
Hence, not only each non-Markovian evolution must have a non-CPTP intermediate map for time intervals starting immediately after $\TL$ (Lemma~\ref{lemma4}), but whenever $\TL<t^\Lambda$ there is a whole continuum of non-CPTP intermediate maps $V_{t^\Lambda,s}$ for $s\in(\TL,t^\Lambda)$. Additionally, if  $V_{t,\TL+\epsilon}$ is not CPTP  for $t>t^\Lambda$, all the intermediate maps $V_{t,s}$ with $s\in(\TL,t^\Lambda)$ are non-CPTP.

In case of $\TL=t^\Lambda$, the infinitesimal intermediate maps $V_{t+\epsilon,t}$ are non-CPTP either for all the times inside a time interval of the type $(\TL, T)$ for some $T>\TL$ or for infinite times that do not constitute an interval for any $T>\TL$. We propose an example of the latter pathological case in Appendix~\ref{lemma123}. A special case is given by the eternal NM model, which has non-CPTP intermediate maps $V_{t,s}$ \textit{for all} $0<s<t$  (see Section~\ref{eternality}): $\TL=t^\Lambda$ and it enjoys both properties described in Proposition~\ref{propDgen}.

\section{Pure non-Markovian evolutions}\label{secPNM}

We show that the initial noise that NNM evolutions apply in the time interval $[0,T^\Lambda]$ is useless for the following non-Markovian effects to happen.
By doing so, we prove that any NNM evolution can be simulated by a Markovian pre-processing of the input states followed by a PNM evolution. Finally, we better explain the role of PNM evolutions and we show that, if an evolution perfectly retrieves the initial information of the system, then it is PNM.
 
\subsection{Simulation of NNM evolutions with PNM evolutions}\label{simulation}

We start by simulating NNM evolutions $\mathbf{\Lambda}$ with the subsequent interaction of the system with two different environments. We consider the Stinespring-Kraus representation theorem~\cite{Stine,Kraus}, which allows us to describe a continuous family of CPTP maps through the interaction of the system with an initially uncorrelated environment. Hence, we consider the system in contact with a first environment $E_1$ in the time interval $[0,\TL]$ and at later times $t>\TL$ in contact with a different environment $E_2$. Thus, consider the following two-step scenario:
\begin{itemize}
\item $t\in[0,T^\Lambda]$ ({\bf Markovian pre-processing}):  the evolution is simulated by the interaction with $E_1$. A unitary transformation $U'_t$ evolves the system-environment state, which at time $0$ is in a product state (no initial system-environment correlations):
\begin{equation}\label{E1}
\rho_{S}(0)\rightarrow \rho_S(t)=\mbox{Tr}_{E_1}[U'_{t} (\rho_S(0)\otimes \omega_{E_1})U'_{t}  ].
\end{equation}
This simulation is possible because $\Lambda_t$ is CPTP for all $t\in [0,\TL]$. The phenomenology during this time interval is Markovian as $\mL$ is CP-divisible (see Eq.~(\ref{TLambdanuovo})). This stage represents the useless noise of $\mL$. 
\item $t=T^\Lambda$: we discard $E_1$ and let the system interact with $E_2$. The system-environment state is given by $\rho_{SE_2}(\TL)=\rho_S(\TL)\otimes \omega_{E_2}$ (no initial system-environment correlations). 
\item $t\geq T^\Lambda$ ({\bf PNM core}): the evolution is simulated by the interaction with $E_2$. A unitary transformation $U''_\tau$ evolves the system-environment state:
\begin{equation}\label{E2}
\rho_{S}(T^\Lambda)\rightarrow \rho_S(t)=\mbox{Tr}_{E_2}[U''_{\tau} (\rho_S(T^\Lambda)\otimes \omega_{E_2})U''_{\tau}  ]  \, ,
\end{equation}
where $\tau=t-T^\Lambda\geq 0$. This simulation is possible because $V_{t,\TL}$ is CPTP for all $t\geq \TL$ (see Eq.~(\ref{TLambdanuovo})).
 The phenomenology during this time interval is NM. 
\end{itemize}

As we already noticed, no information backflow can be observed in $[0,T^\Lambda]$: the phenomenology \textit{in this time interval} is Markovian. 
Now, thanks to this two-stage simulation, we can state that the information involved in the backflows was originally lost later than $T^\Lambda$. Indeed, the non-Markovian effects of $\mathbf{\Lambda}$ do not depend on the behaviour of the dynamics in the time interval $[0,T^\Lambda]$, when $\mL$ is CP-divisible.
This is the reason why we call the first stage a \textit{Markovian pre-processing} and we say that $\Lambda_t$, for $t\in [0,\TL]$, generates the \textit{useless noise} of the evolution. 

Conversely, the (CPTP) intermediate maps $V_{t,s}$ for $\TL\leq s\leq t$ generate the \textit{essential noise} needed for non-Markovian phenomena.
We define $\overline{\mathbf{\Lambda}}=\{\overline{{\Lambda}}_\tau\}_{\tau\geq 0}$ to be the evolution that represents the interaction with the second environment, where:
\begin{equation}\label{E22}
\overline{{\Lambda}}_\tau (\cdot) = \mbox{Tr}_{E_2}[U''_{\tau} (\,\cdot \,\otimes \omega_{E_2} )U''_{\tau}  ] \, .
\end{equation}
 From the above definitions it is easy to see that the dynamical and intermediate maps of $\omL$ are connected with the intermediate maps of $\mathbf{\Lambda}$ as follows:
 \begin{equation}\label{PNM}
\begin{array}{ccc}
& &\,\,\,\,\, \overline{{\Lambda}}_{t}  = V_{t+\TL,\TL}\hspace{0.4cm}   \\ 
& & \overline{V}_{t,s} = V_{t+\TL,s+\TL}
\end{array}  \,\,\,\, \mbox{ for } \,\,\,\,  0 \leq s \leq t \, ,
\end{equation}
where $\overline V_{t,s}$ is the intermediate map of $\omL$.
The map $\overline{{\Lambda}}_{t} $ is CPTP for all $t \geq 0$ (see Eq.~(\ref{TLambdanuovo})) and therefore $\overline{\mathbf{\Lambda}}$ is a valid evolution by itself. It is straightforward to check that  $T^{\overline{\Lambda}}=0$: $\overline{\mathbf{\Lambda}}$ is PNM and we call it the \textit{PNM core of} $\mathbf{\Lambda}$.
Finally, the relation between the characteristic times of $\mL$ and $\omL$ is (see  Eqs.~(\ref{taunm}), (\ref{tnm}) and (\ref{PNM})):
\begin{equation}\label{tLoL}
T^{\overline{\Lambda}}=0 \,,\,\,\,\, \tau^{\overline{\Lambda}}=\tau^\Lambda-\TL \, , \,\,\,\, t^{\overline{\Lambda}}=t^\Lambda-\TL \, .
\end{equation}

We conclude that NNM evolutions can be simulated via a {Markovian pre-processing} (physically represented by Eq.~(\ref{E1})) followed by the action of the corresponding PNM core (physically represented by Eq.~(\ref{E22})):
\begin{equation}\label{NNMevo2}
 \mathbf{\Lambda}= \left\{
\begin{array}{ccc}
 \Lambda_t & t< T^{\Lambda} & \mbox{ (Markovian pre-proc.)} \\
 \overline{\Lambda}_{t-T^\Lambda} \circ \Lambda_{T^\Lambda} & t\geq  T^{\Lambda}  & \mbox{ (PNM core)}
\end{array} \right.  
\end{equation}
This decomposition is depicted in Fig. \ref{resume}.
Naturally, different Markovian pre-processings of the same PNM core provide different NNM evolutions. We summarize the proven relations between NNM and PNM evolutions as follows:
\begin{proposition}\label{PNMNNM}
 Any NNM evolution can be written as a Markovian pre-processing of a PNM evolution.
 Any PNM evolution  defines a class of NNM evolutions given by all its possible  Markovian pre-processings.
\end{proposition}

We conclude this section by discussing the non-Markovian effects of NNM evolutions $\mL$ and their corresponding PNM cores $\omL$.
Differently from before, consider both evolutions acting from the initial time.
Any intermediate map of $\mL$ taking place later than $\TL$ is also present in the dynamics of $\omL$: the intermediate map of $\mL$ during $[s,t]$ is the same intermediate map of $\omL$ during $[s-T^\Lambda,t-T^\Lambda]$, namely $V_{t,s}=\overline V_{t-T^\Lambda,s-T^\Lambda}$. Hence, the two evolutions have the same non-CPTP intermediate maps. It follows that
the non-Markovian effects observable with $\mL$ can also be observed with $\omL$. Nonetheless, the states evolved by $\mathbf\Lambda$ receive a mitigated version of the non-Markovian effects induced by $\overline{\mathbf{\Lambda}}$, where the total attenuation is represented by $\Lambda_{T^\Lambda}$. Later, we show more precisely how this damping acts on information backflows (Section~\ref{tracedistance}) and non-Markovianity measures (Section~\ref{measures}).

\subsection{Features of PNM evolutions}\label{secfeatures}

In this section we study the differences between PNM and NNM evolutions. 
First, we remind that an information backflow can be obtained in a time interval if and only if the corresponding intermediate map is not CPTP~\cite{BD,DDSMJ}. Hence, 
 Lemma~\ref{lemma4} and the two-step simulation of NNM evolutions imply that: \textit{``The information that a quantum system loses with a NNM evolution $\mL$ during the initial time interval $[0,T^\Lambda]$ is never retrieved''}.   
Instead, if we restrict Proposition~\ref{propDgen} to PNM evolutions, it reads:
\begin{corollary}\label{corollaryPNM}
An evolution $\mathbf{\Lambda}$ is PNM if and only if for small enough initial times $\epsilon >0$ there always exists a final time $t>0$ such that $V_{t,\epsilon}$ is non-CPTP. Moreover, $0<t^\Lambda$ implies that $V_{t^\Lambda,s}$  is non-CPTP for all initial times $s\in (0,t^\Lambda)$. Instead, if $0=t^\Lambda$, for all $T>0$ the infinitesimal intermediate map $V_{t+\epsilon,t}$ is non-CPTP for an infinite number of times $t\in(0,T)$.
\end{corollary}
We can say that \textit{``The information that is initially lost with a PNM evolution always takes part in later backflows''}. Indeed, as soon as a PNM evolution starts, even by considering infinitesimal initial times $\epsilon>0$, we have at least one later time $T$ such that there is an information backflow in $[\epsilon,T]$.

We move our attention to an interesting class of evolutions, namely those allowing a \textit{complete information retrieval}. More precisely,  they  first cause a partial or total degradation of the information encoded in the system and later, thanks to one or more backflows, they provide a perfect recovery of the initial information content of the system.
These instances are represented by those $\mathbf{\Lambda}$ having a non-unitary dynamical map $\Lambda_s$ at time $s$ and a unitary dynamical map $\Lambda_t=U$  at a later time $t$.
Since unitary transformations do not degrade the information content of quantum systems, all these evolutions completely retrieve, in the time interval $[s,t]$, any type of information lost in the time interval $[0,s]$. These evolutions are always PNM, where $V_{t,s}$ is not even positive trace-preserving (PTP) (proof in Appendix \ref{proofpropC}):
\begin{proposition}\label{propC}
An evolution characterized by $s<t$ such that $\Lambda_{s}$ is not unitary and $\Lambda_{t}$ is unitary is PNM with the intermediate map $V_{t,s}$ not even PTP.
\end{proposition}

We end this section by clarifying the nature of PNM evolutions via some examples. 
First, one may think that evolutions being first contractive and then unitary are exotic. Interestingly, many of the well-known NM models have PNM cores with this property, e.g., dephasing, depolarizing and amplitude damping channels.

Secondly, having a non-unitary map at time $s$ and a unitary map at time $t>s$  is sufficient but not necessary for an evolution to be PNM. Hence, not all PNM evolutions present complete information retrieval. Indeed, there exist P-divisible PNM  evolutions (see Section~\ref{eternality}).

As a consequence, we might conclude that non-P-divisible PNM evolutions satisfy Proposition~\ref{propC}. The following counterexample proves that this is false. Consider the qubit dephasing map: $\Lambda_{i}(\rho)=p_i \rho+(1-p_i ) \sigma_i \rho \sigma_i$, where $\sigma_{i=x,y,z}$ are the Pauli operators and $p_i \in[0,1]$. Take an evolution such that, for $0<t_1<t_2<t_3$: (i) $\Lambda_{t_1}=\Lambda_{x}$ with $p_x<1$, (ii) $V_{t_2,t_1}=\Lambda_{z}$ with $p_z <1$ and (iii) $V_{t_3,t_2}=\Lambda_{x}^{-1}$, which  is not even PTP. Moreover, we require $\Lambda_{t}=p_x(t) \rho+(1-p_x(t)) \sigma_x \rho \sigma_x$ where $p_x(t)$ is a continuous function such that $p_x(t)<1$ and decreasing in $t\in(0,t_1)$. It is easy to see that this evolution is PNM, e.g., by evolving $\rho(0)=\ketbra{0}{0}$. Indeed, $\Lambda_{t_3}(\ketbra{0}{0})=\ketbra{0}{0}$  but $V_{t_3,\epsilon}(\ketbra{0}{0})$ is outside the state space, which proves that  $V_{t_3,\epsilon}$ is non-CPTP for infinitesimal $\epsilon>0$. Nonetheless, even if at time $t_3$ the evolved state goes back to its initialization ($\rho(0)=\rho(t_3)$), we have that  $\Lambda_{t_3}=\Lambda_{x}^{-1} \circ\Lambda_{z}\circ \Lambda_{x}=\Lambda_{z}$ is not unitary. In this case, the PNM $\mL$ completely retrieves the information content of the system only if properly initialized.

Remarkably, the initial Markovian pre-processing (useless noise) of NNM evolutions may not affect some types of information contained in specific initializations for which the information is instead completely restored. 
Now, we describe a NNM evolution $\mL$ that completely recovers the information to distinguish two states.
Consider the elements of the previous example, where now $\Lambda_{t_1}=\Lambda_{z}$, $V_{t_2,t_1}=\Lambda_{x}$ and $V_{t_3,t_2}=\Lambda_x^{-1}$. The initial states $\ketbra{0}{0}$ and $\ketbra{1}{1}$ get closer during the time interval $[t_1,t_2]$ and later they recover their (maximal) initial distance during the time interval $[t_2,t_3]$. 
Nonetheless, this evolution is NNM because the initial noise $\Lambda_{t}$ for $t\in[0,t_1]$ is useless ($T^{\Lambda}=t_1$). 
Indeed, the initial Markovian pre-processing $\Lambda_{\TL}=\Lambda_{t_1}=\Lambda_z$, although it reduces the distance of several pairs of states, e.g., $\ketbra{+}{+}$ and $\ketbra{-}{-}$, leaves $\ketbra{0}{0}$ and $\ketbra{1}{1}$ untouched. 
Hence, given a NNM evolution $\mL$, the corresponding Markovian pre-processing may not affect the information content of some initializations.
Nonetheless, the corresponding PNM core typically provides larger backflows for generic initializations. For instance, the PNM core $\overline{\mathbf{\Lambda}}$ of the previous example satisfies Proposition~\ref{propC}: it corresponds to the identity map: $\overline{\Lambda}_{t}=I_S$ at $t=t_3-t_1$.

\section{Non-divisible evolutions}\label{nondivisible}

We briefly approach the case of non-divisible evolutions, namely those for which the intermediate map $V_{t,s}$ cannot be written for all $0< s< t$. Great part of the results presented in this work are connected with $\TL$, which in turn is strictly connected with the properties of $V_{t,s}$. Hence, studying the PNM core of a non-divisible NNM evolution sounds problematic.  

We start by reminding that we are considering continuous evolutions (see Section~\ref{evolutions}). Moreover, continuous non-divisible evolutions must have an initial time interval of bijectivity $[0,T^{NB})$ when an inverse $\Lambda_t^{-1}$ exists~\cite{continuity}. Hence, we can consider intermediate maps of the form $V_{t,s}=\Lambda_t \circ \Lambda_s^{-1}$ for all $s<  T^{NB}$. Notice that this is possible for all final times $t$.  Moreover, we remember that invertibility implies divisibility, but the inverse is not true. Therefore, any non-divisible evolution is characterized by a time $T^{ND}$, in general larger than $T^{NB}$, such that $V_{t,s}$ can be defined for all $s < T^{ND}$. As a result, even for non-invertible evolutions there is always a finite time interval inside which we can look for $\TL$. We replace 
 Eq.~(\ref{TLambdanuovo}) with
\begin{equation}\label{LASTEQ}
\!\!\! T^{\Lambda}= \max \left\{\,T\, \left|   
\begin{array}{ccc}
\!\mbox{(A)} &\!\! V_{t,s} \,\,\mbox{ CPTP for all} &\!\!\!  s\leq t \leq T\\ 
\!\mbox{(B)} &\!\! V_{t,T} \,\,\mbox{ CPTP for all} &\!\!\! T\leq t  \\ 
\!\mbox{(C)} &\!\! \Lambda_T \,\, \mbox{ not unitary} &\!\!\! { T>0} \\
\!\mbox{(D)}  & T\leq T^{ND}  &
\end{array}
\right. \!\!\! \right\}  .
\end{equation}
where $T^{ND}=\infty$ for divisible evolutions. Hence, we can proceed with the PNM core extraction introduced in Section \ref{secPNM} whenever $T^\Lambda <   T^{ND}$.

An example where we  obtain the PNM core of a non-invertible NNM depolarizing evolution can be found in Appendix~\ref{noninvdepo}. This example should convince the reader that, for most of the non-divisible dynamics studied in the literature, which usually are almost always divisible \cite{bogna,DDSlong}, the evaluation of $T^\Lambda$  requires the same computational effort needed for divisible dynamics.

\section{Distinguishability backflows}\label{tracedistance}

In this section we study the relation between NNM and PNM evolutions under the point of view of information backflows, as measured by the distinguishability between pairs of evolving states. 
Hence, we analyse the potential of non-Markovian evolutions to make two states more distinguishable in a time interval when the corresponding intermediate map is not CPTP. 

Consider the scenario where we are given one state chosen randomly between $\rho_1$ and $\rho_2$, and we have to guess which state we received. The maximum probability to correctly distinguish the two states through quantum measurements is called  \textit{guessing probability} and it corresponds to
$
P_g(\rho_1,\rho_2)=(2+||\rho_1-\rho_2||_1)/4 
$, 
where $||\cdot||_1$ is the trace norm.  The maximum value 1 is obtained for perfectly distinguishable (orthogonal) states. Instead, the minimum value $1/2$ is obtained if and only if $\rho_1$ and $\rho_2$ are identical. For the sake of simplicity, we define $||\rho_1-\rho_2||_1$ to be the \textit{distinguishability} of $\rho_1$ and $\rho_2$.

Consider a composite system $SA$ with Hilbert space $S (\mathcal{H}_S\otimes\mathcal{H}_A)$, where $S$ is evolved by $\mathbf{\Lambda}$ and $A$ is an ancillary system. Hence, a generic initialization  $\rho_{SA}(0)\in S (\mathcal{H}_S\otimes\mathcal{H}_A)$ evolves as $\rho_{SA}(t)=\Lambda_t\otimes I_A(\rho_{SA}(0))$. Take two states $\rho_{SA,1}(t) $ and $ \rho_{SA,2}(t)$ evolving under the same evolution.
Any increase of $||\rho_{SA,1}(t) - \rho_{SA,2}(t)||_1$ represents a recovery of the missing information needed to distinguish the two states and is a signature of non-Markovianity~\cite{bogna,BLP}. Indeed, this quantity is contractive under quantum channels: CPTP intermediate maps $V_{t,s}$ imply a distinguishability degradation $||\rho_{SA,1}(s) - \rho_{SA,2}(s)||_1 \geq ||\rho_{SA,1}(t) - \rho_{SA,2}(t)||_1$, while a \textit{distinguishability backflow} $||\rho_{SA,1}(s) - \rho_{SA,2}(s)||_1< ||\rho_{SA,1}(t) - \rho_{SA,2}(t)||_1$ implies that $V_{t,s}$ is not CPTP. For this reason, Markovian evolutions are characterized by monotonically decreasing distinguishabilities, while non-Markovian evolutions can provide distinguishability backflows. 
We remember that for all bijective (or almost-always bijective) evolutions there exists a constructive method for an initial pair $\rho_{SA,1}(0)$, $ \rho_{SA,2}(0)$ that provides a distinguishability backflow in $[s,t]$ if and only if the corresponding intermediate map is not CPTP~\cite{bogna}. 

We proceed by studying in which cases and to what extent NNM evolutions damp distinguishability backflows if compared with their corresponding PNM cores. We saw that the initial noise $\Lambda_{\TL}$ of NNM evolutions is useless for non-Markovian phenomena. Now, we quantify how much $\Lambda_{\TL}$ suppresses backflows for each specific initialization.

\begin{proposition}\label{propbackflows}
Consider a NNM evolution $\mathbf{\Lambda}$ providing a  distinguishability backflow in $[s,t]$ by evolving $\rho_{SA,1}(0)$ and $\rho_{SA,2}(0)$, where $V_{t,s}$ is not CPTP. If $\rho_{SA,1}(T^\Lambda)$ and $\rho_{SA,2}(T^\Lambda)$ are not orthogonal, the corresponding PNM core $\overline{\mathbf{\Lambda}}$ provides a larger backflow in $[s-T^\Lambda, t-T^\Lambda]$ by evolving a pair of orthogonal states, where $V_{t,s}=\overline V_{t-T^\Lambda,s-T^\Lambda}$. The proportionality factor between the backflows is $2/||\rho_{SA,1}(T^\Lambda)-\rho_{SA,2}(T^\Lambda)||_1>1$.
\end{proposition}
\begin{proof} We define $\Delta(T^\Lambda)=\rho_{SA,1}(T^\Lambda)- \rho_{SA,2}(T^\Lambda)$, which is hermitian and traceless at all times. This operator corresponds to a segment in the state space with direction and length (as measured by $||\cdot ||_1$) defined by $\rho_{SA,1}(T^\Lambda)$ and $ \rho_{SA,2}(T^\Lambda)$. 
Since $\rho_{SA,1}(\TL)$ and $ \rho_{SA,2}(\TL)$ are not orthogonal, $||\Delta(\TL)||_1<2$. We define $\overline{\Delta}(0)=2\Delta(t)/||\Delta(\TL)||_1$, namely a segment with the same direction of $\Delta(\TL)$ but with $||\overline{\Delta}(0)||_1=2$. Since $\overline{\Delta}(0)$ is hermitian and traceless, it can be diagonalized with  a unitary $ U$, namely  $\overline{\Delta}(0)= U \overline{\Delta}_D(0)  U^\dagger$, where  $\overline{\Delta}_D(0)=\mbox{diag}(\delta_1,\delta_2,\cdots,\delta_d)$, $\delta_i \in [-1,1]$ and $d=d_Sd_A$ is the dimension of $\mathcal{H}_S\otimes \mathcal{H}_A$. Moreover, we have that $||\overline{\Delta}(0)||_1=||\overline{\Delta}_D(0)||_1=\sum_i |\delta_i|=2$ and $\trr{\overline{\Delta}(0)}=\trr{\overline{\Delta}_D(0)}= \sum_i \delta_i=0$. Therefore, we can write $\overline{\Delta}_D(0)=\sigma^+-\sigma^-$, where $\sigma^+$ ($\sigma^-$) is a diagonal density matrix obtained by replacing the negative diagonal elements of $\overline{\Delta}_D(0)$ ($-\overline{\Delta}_D(0)$) with a zero. We define $\overline\rho_{SA,1}(0)=U\sigma_1U^\dagger$ and $\overline\rho_{SA,2}(0)=U\sigma_2U^\dagger$. Notice that $\overline\rho_{SA,1}(0)-\overline\rho_{SA,2}(0)=\overline\Delta(0)$: the two states $\overline\rho_{SA,1}(0)$ and $\overline\rho_{SA,2}(0)$ are orthogonal and their difference $\overline\Delta(0)$ is proportional to $\Delta(\TL)=\rho_{SA,1}(\TL)- \rho_{SA,2}(\TL)$. Hence, if $\mathbf{\Lambda}$ provides a distinguishability backflow in $[s,t]$ by evolving $\rho_{SA,1}(0)$ and $ \rho_{SA,2}(0) $, $\omL$ provides a (larger) backflow in $[s-T^\Lambda,t -T^\Lambda]$ by evolving $\overline\rho_{SA,1}(0)$ and $ \overline\rho_{SA,2}(0) $, where the latter backflow is larger than the former by the factor $2/||\Delta(\TL)||_1>1$. %there exists two pure states $\overline\rho_{SA,1}$ and $\overline\rho_{SA,2}$ such that $\overline\rho_{SA,1} - \overline\rho_{SA,2} =\lambda \Delta (\TL)$, where $\lambda>1$. Hence, we use $\overline\rho_{SA,1}$ and $\overline\rho_{SA,2}$ as inputs for $\omL$ and we define $\overline \Delta(t)=\overline\rho_{SA,1}(t)-\overline\rho_{SA,2}(t)=\overline{\Lambda}_t\otimes I_A(\overline\rho_{SA,1}-\overline\rho_{SA,2})$. Notice that $\overline \Delta(0) = \lambda \Delta(\TL)$: the two hermitian and traceless operators $\overline \Delta(0)$ and $\Delta(\TL)$ are two segments with same direction and different lengths, namely $||\overline \Delta(0)||_1=\lambda ||\Delta(\TL)||_1>||\Delta(\TL)||_1$. 
%Now, it is easy to check that $||\overline \Delta(t-\TL)||_1 - ||\overline \Delta(s-\TL)||_1 =\lambda (||\Delta(t)||_1 - ||\Delta(s)||_1)> ||\Delta(t)||_1 - ||\Delta(s)||_1$. Indeed, $\overline\Delta(s-\TL) =\lambda \Delta(s)$ and $\Delta(t-\TL)=V_{t-\TL,s-\TL} (\overline\Delta(s-\TL))= V_{t-\TL,s-\TL} (\lambda \Delta (s))= \lambda V_{t,s} (\Delta(s))=\lambda \Delta(t)$. 
\end{proof}

Notice that in this proof we provide a constructive method to build the  pairs of states $\overline{\rho}_{SA,1}(0)$, $\overline \rho_{SA,2}(0)$ which, if evolved with the corresponding PNM core, provide larger backflows.
The meaning of this proposition matches the information-theoretic interpretation of PNM evolutions gave above. Indeed, consider a NNM evolution $\mL$ and those pairs of states that provide the largest distinguishability backflows in a time interval $[s,t]$. It can be proven that the corresponding initial states $\rho_{SA,1}(0)$, $\rho_{SA,2}(0)$ are orthogonal. If at time $\TL$  the two states $\rho_{SA,1}(\TL)$ and $\rho_{SA,2}(\TL)$ are still orthogonal when evolved by $\mL$, it means that the noise $\Lambda_\TL$ did not dissipate the information to distinguish this particular pair of states, as discussed in  Section~\ref{secfeatures}. On the contrary, if $\rho_{SA,1}(\TL)$ and $\rho_{SA,2}(\TL)$ are no more orthogonal, $\Lambda_\TL$ dissipated part of the information useful to distinguish the two states unnecessarily: this information does not take part to later backflows. Hence, in these cases the corresponding PNM cores provide larger backflows. 

The results of Proposition~\ref{propbackflows} are independent from $s$, $t$ and the magnitude of the corresponding distinguishability backflow. Indeed,  it solely depends on the information lost after the Markovian pre-processing, namely the distinguishability at time $\TL$. Hence, the results of Proposition~\ref{propbackflows} can be directly extended to all the backflows that the same pair shows, even without knowing their magnitude and when they take place 
\begin{corollary}\label{corpropbackflows}
Consider a NNM evolution $\mathbf{\Lambda}$ and its corresponding PNM core $\omL$. For any pair of states $\rho_{SA,1}(0)$, $\rho_{SA,2}(0)$ that   are not orthogonal at time  $\TL$ and provide one or more distinguishability backflows when evolved by $\mathbf{\Lambda}$, there exists a corresponding pair of orthogonal states such that, if evolved by $\omL$, each backflow is larger by a factor $2/||\rho_{SA,1}(T^\Lambda)-\rho_{SA,2}(T^\Lambda)||_1>1$. The intermediate maps generating the backflows of the two evolutions are the same and the corresponding time intervals differ by a $\TL$ shift.
\end{corollary}

 \section{Non-Markovianity measures}\label{measures}
 
A non-Markovianity measure $M(\mathbf{\Lambda})$ quantifies  the non-Markovian content of evolutions, where Markovianity implies $M(\mL)=0$, while $M(\mL)>0$ implies  $\mL$ to be non-Markovian. As we see below, those measures that are connected with the actual time evolution of one or more states are  influenced by the initial noisy action that precedes non-Markovian phenomena. Hence, in these cases PNM cores provide higher non-Markovianity measures $M(\mathbf{\Lambda})\leq M(\overline{\mathbf{\Lambda}})$: the largest values that any Markovianity measure of this type can assume can be obtained with PNM evolutions and any value assumed with a NNM evolution can be matched or outperformed by a PNM evolution.
The main representatives of this class of measures are defined through the collection of the information backflows obtainable with $\mL$, where this quantity is maximized with respect all the possible system initializations~\cite{BLP,Luo,DDSMJ,WWTAWWTANM}. In the following, we refer to these cases as flux measures. Note that there are non-Markovianity measure satisfying $M(\mathbf{\Lambda})\leq M(\overline{\mathbf{\Lambda}})$ while not being flux measures, e.g.,~\cite{DDSV}.

On the contrary, those measures $M$ that solely depend on the features of intermediate maps, without considering the action of the preceding dynamics, imply $M(\mathbf{\Lambda})= M(\overline{\mathbf{\Lambda}})$. Indeed, a NNM evolution $\mathbf{\Lambda}$ and its corresponding PNM evolution $\overline{\mathbf{\Lambda}}$ have the same non-CPTP intermediate maps. The main representatives of this second class are the Rivas-Huelga-Plenio~\cite{RHP} measure and the $k$-divisibility hierarchy~\cite{SAB}.  We underline that, while flux measures represent the amplitude of phenomena that can be observed, this is not true for this class. 
 
 \subsection{Flux measures}
Flux measures quantify the non-Markovian content of evolutions as follow. Pick an information quantifier and maximize the sum of all the corresponding backflows that $\mL$ shows with respect to all the possible initializations. More precisely, consider a functional $W(\rho_{SA}(t))=W(\Lambda_t\otimes I_A(\rho_{SA}))\geq 0$ which represents the amount of information as measured by $W$ contained in the evolving state. We can also consider quantifiers with multiple input states , e.g., state distinguishability. In order to consider $W$ an information quantifier for $S$, we require it to be contractive under quantum channels on $S$, namely $W(\rho_{SA}(0))\geq  W(\Lambda\otimes I_A (\rho_{SA}))$ for all $\rho_{SA}$ and CPTP $\Lambda$. 

We define the \textit{information flux} as
\begin{equation}\label{flux}
\sigma(\Lambda_t\otimes I_A(\rho_{SA}))=\frac{d}{dt} W (\Lambda_t\otimes I_A(\rho_{SA})) \, .
\end{equation}
Since Markovianity corresponds to CP-divisibility, Markovian evolutions imply non-positive fluxes. Instead, if  $\sigma(\Lambda_t\otimes I_A(\rho_{SA}))>0$, we say that the evolution of $\rho_{SA}$ \textit{witnesses} the non-Markovian nature of $\mathbf{\Lambda}$ through a backflow of $W$.

Flux measures consist of the greatest amount of $W$ that an evolution can retrieve during the evolution with respect to any initialization, namely:
\begin{equation}\label{NW}
M^W(\mathbf{\Lambda})=\max_{\rho_{SA}} \int_{\substack{  \sigma>0}} \sigma_t(\Lambda_t\otimes I_A(\rho_{SA})) dt \, ,
\end{equation}
where the maximization is performed over the whole system-ancilla state space. 

Being $\overline{\mathbf{\Lambda}}$ the PNM core of $\mathbf{\Lambda}$ and $\mbox{Im} (\Lambda_t\otimes I_A)$ the image of the $\mL$ at time $t$, namely those system-ancilla states that can be obtained as an output of $\Lambda_t\otimes I_A$, we obtain:
\begin{eqnarray}\nonumber
&& \!\!\!\!\!\!\! M^W(\mathbf{\Lambda}) = \max_{\rho_{SA}} \int_{\substack{t\geq T^\Lambda \\  \sigma>0}} \sigma (\Lambda_t\otimes I_A (\rho_{SA})) \,dt \\ \nonumber
&&\!\!\!\!\!= \max_{\rho_{SA}\in \mbox{\small Im}(\Lambda_{T^\Lambda} \otimes I_A)} \int_{\substack{t\geq T^\Lambda \\  \sigma>0}}\sigma (V_{t,T^{\Lambda}}\otimes I_A (\rho_{SA}))  \,dt
\\ \nonumber
&&\!\!\!\!\! = \max_{ \rho_{SA}\in \mbox{\small Im}(\Lambda_{T^\Lambda} \otimes I_A)} \int_{\substack{  \sigma>0}} \sigma(\overline{\Lambda}_t \otimes I_A (\rho_{SA})) \,dt \\ \label{diffN}
 &&\!\!\!\!\!\leq \max_{ \rho_{SA}} \!\int_{\substack{ \sigma>0}} \sigma(\overline{\Lambda}_t \otimes I_A (\rho_{SA}))\,dt = M^W(\overline{\mathbf{\Lambda}})  \, ,
\end{eqnarray}
where the first equality is justified by the fact that backflows can only happen for $t\geq T^\Lambda$ (any NNM evolution is CP-divisible in $[0,T^\Lambda]$), the second equality is a simple consequence of $\Lambda_t=V_{t,T^\Lambda}\circ \Lambda_{T^\Lambda}$, the third equality follows from Eq.~(\ref{PNM}) and the inequality follows from the enlargement of the maximization space. 
 
It is interesting to understand when we can obtain $M^W(\mathbf{\Lambda})< M^W(\omL)$. 
Consider the information quantifier $D(\rho_{SA,1}(t),\rho_{SA,2}(t))=||\rho_{SA,1}(t)-\rho_{SA,2}(t)||_1$ for a fixed ancilla $A$, where in this case we consider the evolution $\mL$. We call $\{\rho^{i}_{SA,1},\rho^{i}_{SA,2}\}$ those pairs that allow to obtain the maximum of Eq.~(\ref{NW}).   Notice that these pairs of states are always initially orthogonal: $D(\rho^{i}_{SA,1}(0),\rho^{i}_{SA,2}(0))=2$ for all $i$.  Hence, if $\sigma^D(\Lambda_t\otimes I_A (\rho_{SA,1}),\Lambda_t\otimes I_A (\rho_{SA,2})) = \sigma^D(\Lambda_t\otimes I_A ,\rho_{SA,1},\rho_{SA,2})$ is the flux associated to $D(\rho_{SA,1}(t),\rho_{SA,2}(t))$ as in Eq.~(\ref{flux}), we have:
\begin{eqnarray}\nonumber
&&\!\!\!\!\!\!\!\!\!\!\!\!\!\!\!M^D(\mL)= \max\limits_{\substack{\rho_{SA,1} \\ \rho_{SA,2}}} \!\int_{
\substack{ \sigma^D>0}
} \sigma^D(\Lambda_t\otimes I_A ,\rho_{SA,1},\rho_{SA,2}) \,dt
\\ \label{measureNMdist}
&&\!\!\!\!\!\!\!\!\!\!\!\!\!\!\!=  \int_{
\substack{  \sigma^D>0}} \sigma^D(\Lambda_t\otimes I_A ,\rho_{SA,1}^i,\rho_{SA,2}^i) \,dt  \,\,\, \mbox{for all } i.
\end{eqnarray}
Thanks to Corollary \ref{corpropbackflows}, we can prove that:
\begin{equation} \nonumber
M^D(\mL)\!\leq\! M^D(\omL) \!=\! \max_i \frac{2 M^D(\mL)}{||\rho^{i}_{SA,1}(T^\Lambda)-\rho^{i}_{SA,2}(T^\Lambda)||_1}  \, ,
\end{equation}
where we set $\rho^{i}_{SA,1}(\TL)=\Lambda_\TL\otimes I_A (\rho^{i}_{SA,1}(0))$ and $\rho^{i}_{SA,2}(\TL)=\Lambda_\TL\otimes I_A( \rho^{i}_{SA,2}(0))$.
Hence, if the information content of the pairs $\{\rho^{i}_{SA,1},\rho^{i}_{SA,2}\}$ at time $\TL$ is lower than at the initial time when evolved by $\mL$, then $M^D(\mL)<M^D(\omL)$. Moreover, the proportionality factor between the two measures is given by the states $\{\rho^{j}_{SA,1},\rho^{j}_{SA,2}\}$ that get the closest at time $\TL$. In Section~\ref{secdepo} we explicitly evaluate $M^D(\mL)$ and $M^D(\omL)$ in case of depolarizing evolutions and we show that, even without ancillary systems, $M^D(\mL)<M^D(\omL)$ is always verified.

A second measure similar to $M^W$ is given by~\cite{vault}:
\begin{equation}\label{Nmax}
M^{W,max}(\mathbf{\Lambda})=\max_{s< t,\, \rho_{SA}} \{0,W(\rho_{SA}(t))-W(\rho_{SA}(s))\} , 
\end{equation}
which corresponds to the largest backflow of $W$ that the dynamics is able to show in a single time interval. Finally, a third measure is~\cite{vault}:
\begin{equation}\label{Nav}
M^{W,av}(\mathbf{\Lambda})=\max_{t>0, \rho_{SA}} \{0,  W(\rho_{SA}(t))-\langle W(\rho_{SA}(t))\rangle \} , 
\end{equation}
where $\langle W(\rho_{SA}(t)) \rangle = t^{-1} \int_0^t W(\rho_{SA}(s)) ds $. This measure
 corresponds to the largest difference, with respect to $t$, between the information $W$ at time $t$ and its average in the interval $[0,t]$. Moreover, $M^{W,av}$ has a precise operational meaning connected with the probability to store and faithfully retrieve information by state preparation and measurement, where an attack performed by an eavesdropper may occur. It can be proven~\cite{vault} that, for any $\mL$ and $W$, $M^{W,av} (\mathbf{\Lambda}) \leq M^{W,max} (\mathbf{\Lambda}) \leq M^{W} (\mathbf{\Lambda})$. 
Similarly to Eq.~(\ref{diffN}), it is possible to demonstrate that:
\begin{eqnarray}\nonumber  M^{W,max}(\mathbf{\Lambda}) &\leq&   M^{W,max}(\overline{\mathbf{\Lambda}})  \,  , \\ \nonumber 
M^{W,av}(\mathbf{\Lambda}) &\leq&   M^{W,av }(\overline{\mathbf{\Lambda}}) \, .
\end{eqnarray}

\subsection{Incoherent mixing measure}

	A second type of  non-Markovianity measure corresponds to the minimal incoherent Markovian noise needed to make a non-Markovian evolution $\mL$ Markovian~\cite{DDSV}. In order to describe this measure, we first consider an evolution obtained as convex combination of $\mL$ and a generic Markovian evolution $\mL^M$.
We consider the mixed evolution
$
\mathbf{\Lambda}^{mix}_p=(1-p)\mathbf{\Lambda}+p\mathbf{\Lambda}^M \, 
$. 
and define a non-Markovianity measure by looking for the minimal value of $p$, hence the minimal amount of Markovian noise, such that $\mL^{mix}_p$ is Markovian, namely:
\begin{equation}\label{Nmix}
M^{mix}(\mathbf{\Lambda})=\min\{\,p\,|\, \exists \mathbf{\Lambda}^M \mbox{ s.t. } \mathbf{\Lambda}^{mix}_p \mbox{ is Markovian} \} \, .
\end{equation}
In Appendix \ref{appNmix} we prove that:
\begin{equation}\nonumber
M^{mix}(\mathbf{\Lambda})\leq M^{mix}(\omL) \, .
\end{equation}
Finally, in Section~\ref{secdepo} we show that $M^{mix}(\mathbf{\Lambda})< M^{mix}(\omL)$ for all NNM depolarizing evolutions $\mL$.

\subsection{RHP measure and $k$-divisibility}

Consider a generic PNM evolution $\overline{\mathbf{\Lambda}}$ and all the corresponding NNM evolutions $\mathbf{\Lambda}$ that can be obtained from $\overline{\mathbf{\Lambda}}$ with a Markovian pre-processing. As we saw, $\overline{\mathbf{\Lambda}}$ and all its corresponding $\mL$ have the same non-CPTP intermediate maps. 
Therefore, the non-Markovianity measures that solely depend on the properties of non-CPTP intermediate maps, since they are not influenced by the particular (useless) noise that precedes their action, assume the same value for $\omL$ and all its corresponding $\mL$. This is the case of the RHP measure $\mathcal{I}(\mL)$ (see Eq.~(4) from Ref.~\cite{RHP}), and the $k$-divisibility non-Markovian degree NMD$[\mL]$  (see Ref.~\cite{SAB}):
\begin{equation}
\mathcal{I}(\omL) = \mathcal{I}(\mL) \,\, \, \mbox{ and } \,\,\,  \mbox{NMD}(\omL)= \mbox{NMD}(\mL) \, .
\end{equation}

\section{Entanglement breaking property}\label{EB}

We call $C(\rho_{AB})$ a correlation measure for the bipartite system $AB$ if: (i) $C(\rho_{AB})\geq 0$ for all $\rho_{AB}$, (ii) $C(\rho_{AB})= 0$ for all product states $\rho_{A}\otimes \rho_B$, (iii) $C(\Lambda_A\otimes I_B(\rho_{AB}))\leq C(\rho_{AB})$ and $C(I_A\otimes \Lambda_B(\rho_{AB}))\leq C(\rho_{AB})$ for all $\rho_{AB}$ and CPTP maps $\Lambda_A$ and $\Lambda_B$. 
Entanglement measures, denoted here by $E$, capture only non-classical correlations. Indeed, they satisfy the additional property of being non-increasing under local operations, namely (iii), assisted by classical communication (LOCC). As a consequence, $E(\rho_{AB})=0$ for all separable states, namely those that can be written as statistical mixtures of product states: $\rho_{AB}=\sum_i p_i \rho_A\otimes \rho_B$, where $\{p_i\}_i$ is a probability distribution.

We discuss how the link between NNM and PNM evolutions behaves with respect to the entanglement breaking (EB) property. A quantum channel $\Lambda_S$ is EB if  it destroys the entanglement of any input state, namely if $\Lambda_S \otimes I_A (\rho_{SA})$ is separable for all $\rho_{SA}$. Consider a generic $\mL$. We say that it is EB if  there exists a time $t^{EB,\Lambda}>0$ such that $\Lambda_t$ is EB for all $t\geq t^{EB,\Lambda}$. 

Take a PNM evolution $\omL$ and a NNM evolution $\mL$ that can be obtained with a Markovian pre-processing of $\omL$. This pre-processing cannot increase the amount of entanglement of any state. Hence, if $\omL$ is EB, then $\mL$ is EB. Nonetheless, in case of $\omL$ and $\mL$ EB, there is no general order for the corresponding EB times: $t^{EB,\Lambda}>t^{EB,\overline\Lambda}$ and $t^{EB,\Lambda}<t^{EB,\overline\Lambda}$ are both possible.

We have to keep in mind that, if a generic NNM evolution $\mL$ is EB, we cannot immediately say anything about the EB nature of $\omL$ and we must study the particular dynamics more in detail. 
Indeed, it is easy to find NNM evolutions $\mL$ with EB useless noises $\Lambda_{\TL}$, where the corresponding PNM core $\omL$ is not EB. 
Also, there exist cases where the Markovian pre-processing $\Lambda_{\TL}$ is not EB, the PNM core $\omL$ is not EB, but the corresponding NNM evolution $\mL$ is EB.

\subsection{Activation of correlation backflows}

We now discuss a technique focused on entanglement revivals which can be easily generalized to other correlation measures. 
The isolation of the PNM core of a NNM evolution may lead to the \textit{activation} of entanglement backflows.
Take a bipartite system $SA$, where $S$ is evolved by $\mL$ and $A$ is an ancilla.
Whenever we have a backflow of $E$, the same backflow can also observed with $\omL$, namely the corresponding PNM evolution. Moreover, as we saw for the corresponding flux non-Markovianity measure, $M^{E}(\mL)\leq M^{E}(\omL)$. What is interesting is the possibility to {activate} backflows of entanglement through the isolation of the PNM core, namely when $M^{E}(\mL)=0$ and $M^{E}(\omL)>0$. This scenario is made possible when $\mL$ is EB, the corresponding non-CPTP intermediate maps $V_{t,s}$ take place only for $t^{EB,\Lambda}\leq s<t$ and the corresponding PNM core $\omL$ is not EB. In this case, when an entangled state is evolved by $\mL$ and a non-CPTP intermediate map takes place, all the entanglement has already been destroyed and no backflows are possible. Instead, for a system evolving under $\omL$, when the (same) non-CPTP intermediate map takes place entanglement can be non-zero and backflows are allowed.

Whenever a non-Markovian evolution does not provide correlation backflows, additional ancillary degrees of freedom can activate the possibility to observe backflows. This phenomena has already been studied for entanglement~\cite{Janek,DDSDONATO} and Gaussian steeriing~\cite{DDSDONATO}. For instance, instead of evaluating entanglement among $S$ and $A$, we would need to evaluate it among $SA'$ and $A$, where $A'$ is a second ancilla. Hence, our construction allows a different strategy to obtain correlation backflows in those situations where an $SA$ setup does not show any: instead of implementing additional ancillary systems, which may result in an experimental setup that is more demanding to handle or not even realisable with the available tools, we can simply consider the PNM core of the studied evolution.

\section{Depolarizing model}\label{secdepo}

We apply our results to a simple model called depolarizing. Starting from a generic NNM depolarizing evolution $\mL$, we show how to find $\TL$, $\tau^\Lambda$ and $t^\Lambda$, the corresponding PNM evolution $\omL$ and we calculate the gains in terms of information backflows and non-Markovianity measures that $\omL$ provides with respect to $\mL$. 
We conclude by applying our technique to an explicit toy model. Moreover, we show how our approach can be directly applied to non-bijective depolarizing evolutions in Appendix \ref{noninvdepo}.

We define a generic depolarizing evolution $\mL$  through the corresponding dynamical map, namely
\begin{eqnarray}\nonumber
\Lambda_t (\,\cdot\,) =  f(t) I_S ( \,\cdot\,) + (1-f(t)) \tr{ \,\cdot\, }\frac{\mathbbm{1}_S}{d}  \, ,
\end{eqnarray}
where $d$ is the dimension of the system $S$,
$I_S$ is the identity map and ${\mathbbm{1}_S}/{d}$ is the maximally mixed state~\cite{DDSV}. The behaviour of the evolution is determined by the \textit{characteristic function} $f(t)$. The dynamical maps $\Lambda_t$ are CPTP, continuous in time and such that $\Lambda_0=I_S$ if and only if $f(t)\in [-1/(d^2-1), 1]$ is a continuous function such that $f(0)=1$. 
For the sake of simplicity, from now on we restrict our attention to depolarizing evolutions with $f(t)\in[0,1]$ for all $t\geq 0$. Those cases of $f(t)$ assuming negative values necessitate a simple generalization of the techniques used here. An in-depth analysis of depolarizing evolutions with $f(t)\in [-1/(d^2-1), 1]$ can be found in Ref.~\cite{DDSV}.
 The evolution $\mL$ is invertible if and only if $f(t)>0$ at all times. Indeed,  $f(t^{NB})=0$ implies that every initial state is mapped into the same (maximally mixed) state:  $\Lambda_{t^{NB}}(\rho_S(0))= \mathbbm{1}_S/d$. In this case $\Lambda_{t^{NB}}$ is non-invertible and  $V_{t,t^{NB}}$ cannot be defined.

The interpretation of depolarizing evolutions is straightforward: at time $t$ each state is mixed with the maximally mixed state $\mathbbm{1}_S/d$ with a ratio given by $f(t)$. The larger $f(t)$ is, the closer is $\rho_S(t)$ to the initial state $\rho_S(0)$. Moreover, this contraction towards the maximally mixed state is symmetric in the state space. Indeed, for any two initial states $ \rho_{S,1}(0)$ and $\rho_{S,2}(0)$ evolving under $\mL$ we have:
\begin{eqnarray}\label{tracedepo}
\!\!\!\!\!\!\!\!\!\!\!{||\rho_{S,1}(t)-\rho_{S,2}(t)||_1}= f(t) {||\rho_{S,1}(0)-\rho_{S,2}(0)||_1}  .
\end{eqnarray}

The intermediate map corresponding to the depolarizing evolution during a generic time interval $[s,t]$ assumes the same form of a depolarizing dynamical map, namely:
\begin{eqnarray}\label{depoint}
\!\!\!\!\!\!\! V_{t,s} (\,\cdot\,) = \frac{f(t)}{f(s)} \, I_S (\,\cdot\,) + \left(1-\frac{f(t)}{f(s)} \right) \tr{ \,\cdot\, }\frac{\mathbbm{1}_S}{d} \, .
\end{eqnarray}
Hence, the CPTP condition  for $V_{t,s}$ coincides with $f(t)/f(s) \in [0, 1]$, that is $f(s) \geq f(t)$ for $s\leq t$. Similarly, the infinitesimal intermediate map $V_{t+\epsilon,t}$ is CPTP if and only if $f'(t)\leq 0$. Indeed, Markovian depolarizing evolutions have non-increasing characteristic functions. 

The Choi state of $V_{t,s}$ is $V_{t,s}\otimes I_S (\phi^+_{SA})= f(t)/f(s) \phi^+_{SA} + (1-f(t)/f(s)) \mathbbm{1}_{SA}/d^2$ and its smallest eigenvalue is $\lambda_{t,s}=(1-f(t)/f(s))/d^2$. Since $V_{t,s}$ is CPTP if and only if $\lambda_{t,s}\geq 0$, thanks to the evaluation of $\lambda_{t,s}$ we are able to obtain $\mathcal P^\Lambda$ and $\mathcal N^\Lambda$, the collection of time pairs $\{s,t\}$ such that $V_{t,s}$ is, respectively, CPTP and non-CPTP (see Eqs.~(\ref{CPTPsub}) and (\ref{nonCPTPsub})).

Non-Markovian depolarizing evolutions have non-monotonic characteristic functions.
An increase of $f(t)$ in a given time interval corresponds to a corresponding non-CPTP intermediate map. Moreover, in the same time interval the trace distance between any two states increases, namely a distinguishability backflow.
The largest distinguishability backflows are provided by initially orthogonal states, for which the trace distance is equal to $2f(t)$ (see Eq.~(\ref{tracedepo})). We consider the flux non-Markovianity measure $M^D$ in case of no ancillary systems (see Eq.~(\ref{measureNMdist})): 
\begin{eqnarray}\nonumber
&& \!\!\!\!\!\!\!\!\!\!\!\!\! M^D(\mL) =  \int_{ \sigma^D>0} \sigma^D(\Lambda_t (\rho_{S,1}),\Lambda_t (\rho_{S,2})) dt \\ \nonumber &&\!\!\!\!\!\!\!\!\!\!\,\,= 2 \int_{f'>0} f'(t) dt  = 2 \sum_i f(t_{fin,i})-f(t_{in,i}) = 2 \Delta  \, ,
\end{eqnarray}
where $\rho_{S,1}$, $\rho_{S,2}$ are any two orthogonal states, $(t_{in,i},t_{fin,i})$ is the $i$-th time interval when $f'(t)>0$ and $\Delta>0$ is the sum of all the revivals of $f(t)$. Finally, the non-Markovianity measure given in Eq.~(\ref{Nmix}) is equal to $M^{mix}(\mL)=\Delta/(1+\Delta)$~\cite{DDSV}.

\subsection{Backflows timing and PNM core}\label{backdepo}

 We are ready to evaluate $\TL$, $\tau^\Lambda$ and $t^\Lambda$. We can rewrite Eqs.~(\ref{TLambdanuovo}), (\ref{taunm}) and (\ref{tnm}) in terms of $f(t)$ and $f'(t)$ as follows:
\begin{eqnarray}
&&\!\!\!\!\!\!\!\!\!\!T^{\Lambda} = \max \left\{ T \,  \left|  \! \!\, 
\begin{array}{ccc}
\mbox{(A)} &\!  \! \! f'(t)\leq 0  & \! \! \! \mbox{ for all }\,    t \leq T,\\ 
\mbox{(B)} & \!  \! \!f(T)\geq f(t) & \!  \! \! \,\mbox{ for all } \, T\leq t,  \\ 
\mbox{(C)} &\!   \! \!f(T)\neq 1   \, & \! \! { T>0.}
\end{array}
\right. \!\!\right\} , \nonumber \\ \label{TLambdadepo}
 \\ 
 \label{taudepo}
&&\!\!\!\!\!\!\!\!\!\! \tau^{\Lambda} = {\inf} \left\{ T \, \left| \, f'(T)>0 \right. \right\},
\\ \label{tdepo}
&&\!\!\!\!\!\!\!\!\!\! t^{\Lambda}= \min \left\{\,T \, \left| \, f(T)=f(\TL) \mbox{ for } T>\TL \right. \right\} ,
\end{eqnarray}
where the last equality holds because Eq.~(\ref{TLambdadepo}) implies that $f(\TL)>f(\TL+\epsilon)$ for infinitesimal $\epsilon>0$. As expected, condition (A) of Eq.~(\ref{TLambdadepo}) implies that $\mL$ behaves as a Markovian depolarizing evolution in the time interval $[0,\TL]$. Secondly, by considering (A) and (B) together, we can state that $f(\TL)\in (0,1)$. As discussed in Section~\ref{nondivisible}, in case $\mL$ is non-invertible and $t^{NB}$ is the earliest time when $f(t^{NB})=0$, we should add to Eq.~(\ref{TLambdadepo}) the constraint $\TL<t^{NB}$. Anyway, as we show in Appendix \ref{noninvdepo}, even without imposing such a constraint, $\TL<t^{NB}$.

Generic non-Markovian evolutions are characterized by $0\leq \TL\leq \tau^\Lambda\leq t^\Lambda$ (see  Eq.~(\ref{order})). Nonetheless, depolarizing evolutions always satisfy $0\leq \TL< \tau^\Lambda< t^\Lambda \leq \infty$. The last equality ($t^\Lambda=\infty$) is achieved when $f(\TL)>f(t)$ for all $t>\TL$ and $\lim_{t\rightarrow \infty} f(t)=f(\TL)$.

We obtain the PNM core of a NNM depolarizing evolution by exploiting the method presented in Section~\ref{simulation}.  Hence, if we apply Eq.~(\ref{PNM}) to the intermediate maps of a NNM depolarizing evolution $\mL$ characterized by $f(t)$, we obtain the PNM depolarizing evolution $\omL$ characterized by
\begin{widetext}

\begin{figure}
\begin{center}
\includegraphics[width=0.75\textwidth]{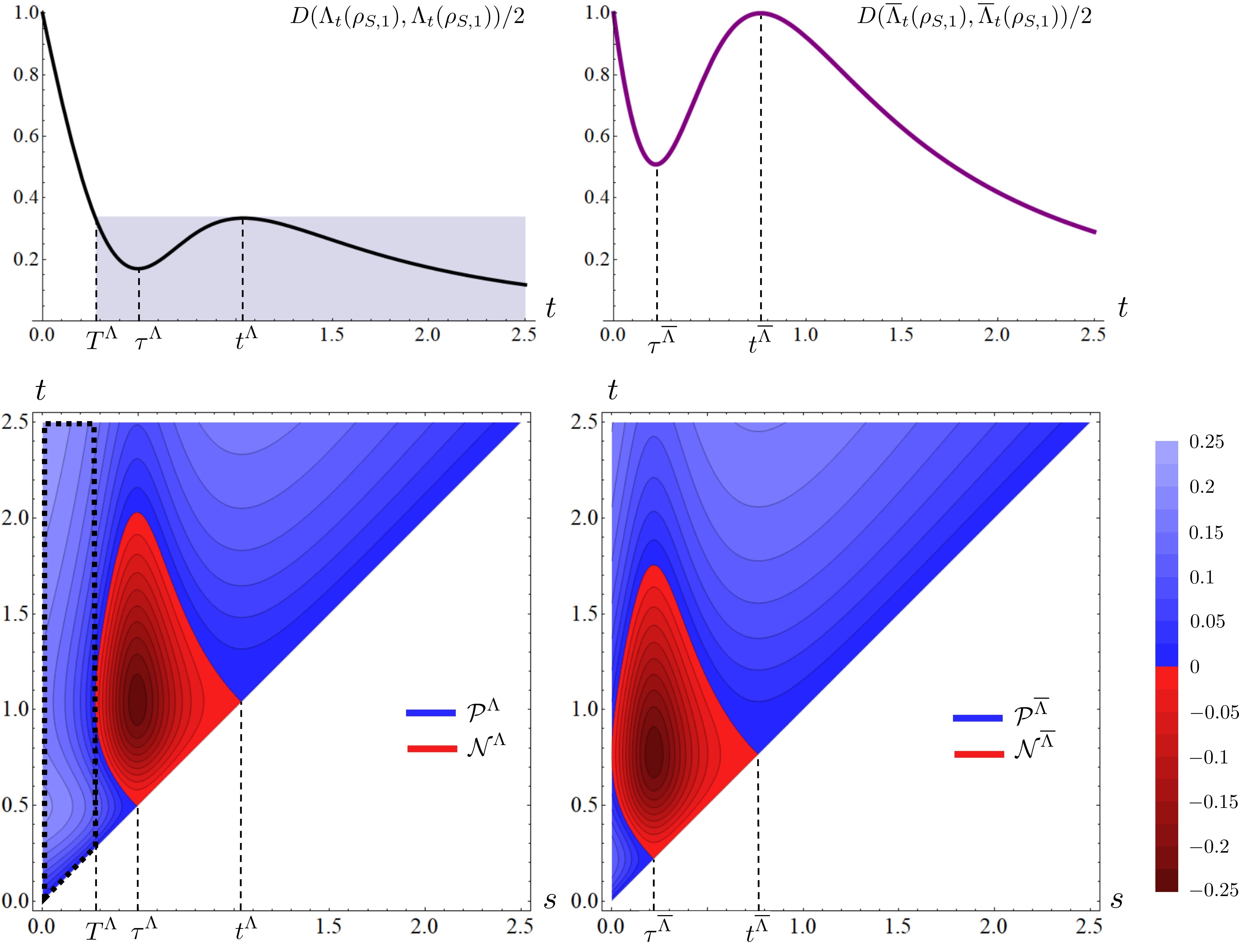} 
\end{center}
\caption{
{\bf Top left:} Trace-distance $D(\Lambda_t(\rho_{S,1}),\Lambda_t(\rho_{S,2}))/2=f(t)$  of two initially orthogonal states evolved by the NNM depolarizing evolution $\mL$ given by $f(t)=(1 - 3 t + 2 t^2 + 2 t^3) / (1 + t^2 + t^3 + 3 t^5)$. 
The evolution in $[0,\TL]$ is a Markovian pre-processing. Indeed, all the states get closer and this noise is not needed for later increases: $f(t)$ cannot increase in time intervals starting before $\TL$.
The NM nature of $f(t)$ is shown immediately after $\TL$: $f(t^\Lambda)-f(s)>0$ for all $s\in(\TL,t^\Lambda)$. The total increase of $f(t)$ is $\Delta=f(t^\Lambda)-f(\tau^\Lambda)\simeq 0.16$.
The first time after which $f(t)$ increases instantaneously is $\tau^\Lambda$.  
The shaded region represents the PNM core damped by the Markovian pre-processing.
{\bf Top right:} $D(\overline\Lambda_t(\rho_{S,1}),\overline\Lambda_t(\rho_{S,2}))/2=\overline f(t)$, where $\rho_{S,1,2}$ are orthogonal and $\omL$ is the PNM core of $\mL$, the depolarizing evolution given by $\overline{f}(t)=f(t+\TL)/f(\TL)$, which increases in time intervals starting immediately after the initial time: $f(t^{\overline\Lambda})-f(s)>0$ for all $s\in (0,t^{\overline\Lambda})$. 
 Since $\overline f(t^{\overline\Lambda})=1$, we have $\overline\Lambda_{t^{\overline\Lambda}}=I_S$: all the initial information is restored. The total increase of $\overline f(t)$ is $\overline\Delta=\overline f(t^{\overline\Lambda})-f(\tau^{\overline\Lambda})=\Delta/f(\TL) \simeq 0.49$.
{\bf Bottom left:} $\mathcal{P}^\Lambda$ and $\mathcal{N}^\Lambda$, the sets of time pairs $\{s,t\}$ when $V_{t,s}$ is respectively CPTP and non-CPTP. Lighter blue (darker red) colours corresponds to larger (larger in modulo)  positive (negative) values of $\lambda_{t,s}=(1-f(t)/f(s))/4$, the lowest eigenvalue of the  Choi operator of $V_{t,s}$. The minimum is obtained for $\lambda_{t^\Lambda,\tau^\Lambda}\simeq-0.241$ ($f(t)$ shows its largest increase in $[\tau^\Lambda, t^\Lambda]$).  Along the $s=t$ line $\lambda_{t,s}=0$ and the adjacent points correspond to the infinitesimal intermediate maps $V_{t+\epsilon,t}$, which are non-CPTP if and only if $f'(t)>0$. Any NNM evolution ($\TL>0$) has $\mathcal P^\Lambda$ that owns all the points  $\{s,t\}$ for $s\leq \TL$ (dotted region): $\TL$ is the largest $T$ such that $\{s,t\}\in \mathcal P^\Lambda$ for all $s\leq T$.
{\bf Bottom right:} $\mathcal{P}^{\overline\Lambda}$ and $\mathcal{N}^{\overline\Lambda}$.
The different shades represent different values of $\overline\lambda_{t,s}$, the lowest eigenvalue of the Choi operator of $\overline V_{t,s}$. Since $\overline V_{t,s}=V_{t+\TL,s+\TL}$, the minimum is $\overline \lambda_{t^{\overline\Lambda},\tau^{\overline\Lambda}}\simeq -0.241$.
}\label{newfigdepo}
\end{figure}
\end{widetext}
$\overline{f}(t)=f(t+\TL)/f(\TL)$. It is a valid characteristic function 
 ($\overline{f}(t)\in [0,1]$ and $\overline{f}(0)=1$) and $T^{\overline\Lambda}=0$. Hence, the the corresponding dynamical maps are:
\begin{eqnarray}\nonumber
&&\!\!\!\!\!\!\!\!\!\! \overline{\Lambda}_t (\,\cdot\,) =  \overline{f}(t) I_S ( \,\cdot\,) + \left(1-\overline{f}(t)\right) \tr{ \,\cdot\, }\frac{\mathbbm{1}_S}{d} 
\\ \nonumber &&\!\!\!\!\!\!\!\!\!\! \,\,\,\,\,= 
\frac{f(t+\TL)}{f(\TL)} I_S ( \,\cdot\,) + \left(1-\frac{f(t+\TL)}{f(\TL)} \right) \tr{ \,\cdot\, }\frac{\mathbbm{1}_S}{d}    \, .
\end{eqnarray}

The NNM evolution $\mL$ can be expressed as a first time interval of Markovian pre-processing, expressed by $\Lambda_t$ for $t\in [0,\TL]$, followed by the action of the PNM evolution $\omL$ (see Eq.~(\ref{NNMevo2})). 
As we explained, $\omL$ is nothing but $\mL$ without the resultant of its Markovian pre-processing $\Lambda_\TL$, which, not only is useless for the appearance of non-Markovian phenomena but damps information backflows.  Indeed, we can apply Corollary \ref{corpropbackflows} and conclude that whenever we can obtain a distinguishability backflow with $\mL$ in a time interval $[s,t]$, we can observe a backflow with $\omL$ in the time interval $[s-\TL,t-\TL]$, where the proportionality factor between the two revivals is $1/f(\TL)>1$. 
%Secondly, we have the amplification of the flux non-Markovianity measure  $M^D(\omL)=M^D(\mL)/f(\TL)>M^D(\mL)$ (see Eq.~(\ref{measureNMdist})). Similary, since the sum of all the revivals of $\overline{f}(t)$ is equal to $\overline{\Delta}=\Delta/f(\TL)>\Delta$, we obtain $M^{mix}(\omL)=\Delta/(f(\TL)+\Delta)> M^{mix}(\mL)$.
As expected, $\omL$ is characterized by larger non-Markovianity measures than $\mL$:
\begin{eqnarray}\label{ND1}
 M^D(\omL)=
 \frac{2\Delta}{f(\TL)}> M^D(\mL) = 2\Delta  \, ,  \hspace{.7cm}
\\
 \label{Nmix1}
M^{mix}(\omL)\, 
=\frac{\Delta}{f(\TL)+\Delta} > M^{mix}(\mL) = \frac{\Delta}{1+\Delta}\, .
\end{eqnarray} 
It can be proven that similar results holds true for the measures $M^{W,max}$ and
$M^{W,av}$ (see Eqs.~(\ref{Nmax}) and (\ref{Nav})).  

We conclude by noticing that all PNM depolarizing evolutions $\omL$ completely retrieve the initial information of the system  at time $t^{\overline \Lambda}$. In particular, all PNM depolarizing evolutions satisfy the conditions of 
Proposition~\ref{propC}, where $\overline{\Lambda}_{t^{\overline \Lambda}}= I_S$. This result follows from the observation that $f(t^\Lambda)= f(\TL)$, and therefore all PNM depolarizing evolutions are such that $\overline f(t^{\overline\Lambda})=\overline f(0)=1$. Notice that $t^{\overline\Lambda}$ may be divergent.

\subsection{Example}\label{newsecexample}
We show how to apply our results to a simple characteristic function $f(t)$ representing a NNM depolarizing evolution $\mL$.
The toy model considered here is given  by $f(t)=(1 - 3 t + 2 t^2 + 2 t^3) / (1 + t^2 + t^3 + 3 t^5) $ (see Figure \ref{newfigdepo}): a continuous function with a single time interval of increase and an infinitesimal asymptotic behaviour.  We start by calculating the times $\TL$, $\tau^\Lambda$ and $t^\Lambda$. Hence, we consider the sets $\mathcal P^\Lambda$ and $\mathcal N^\Lambda$, the sets containing the pairs of times $\{s,t\}$ such that $V_{t,s}$ is respectively CPTP and non-CPTP (see Eqs.~(\ref{CPTPsub}) and~()).  We can obtain these sets by noticing that the smallest eigenvalue $\lambda_{t,s}=(1-f(t)/f(s))/d^2$ of the Choi state of $V_{t,s}$ is non-negative if and only if $V_{t,s}$ is CPTP. The same analysis is performed for the corresponding PNM core $\omL$. 
We start with a technical analysis of $f(t)$. 
Standard numerical methods lead to $\TL\simeq 0.275$, $\tau^\Lambda=0.495$ and $t^\Lambda\simeq1.040$.
It is possible to have increases of $f(t)$ only in time intervals $[s,t]$ starting later than $\TL$. Moreover, as explained by Proposition~\ref{propDgen}, these increases take place for a continuum of initial times:  $f(t^\Lambda)-f(s)>0$ for all $s\in(\TL, t^\Lambda)$.  Instead, if
we consider an initial time $s$ sooner than $\TL$, the characteristic function cannot increase: $f(t)-f(s)<0$ for $s< \TL$ and $s< t$.
The time $\tau^\Lambda$ is the first time after which $f'(t)>0$. Moreover, $f'(t)>0$ only for $t\in(t_{in},t_{fin})=(\tau^\Lambda,t^\Lambda)$, where the total revival is $\Delta=f(t^\Lambda)-f(\tau^\Lambda)\simeq 0.164$.

We now analyse $f(t)$ from the point of view of information backflows. The characteristic function $f(t)$ is directly connected with the time-dependent distinguishability $D(\rho_{S,1}(t),\rho_{S,2}(t))$ of two states evolving under $\mL$ (see Eq.~(\ref{tracedepo})). 
In the first time interval $[0,\TL]$ information is lost and never recovered. Indeed, we called this noise \textit{useless} for non-Markovian phenomena and the resultant noise $\Lambda_{\TL}$ represents a Markovian pre-processing. As discussed above, the damping of the initial Markovian pre-processing is quantified by $f(\TL)\simeq 0.334 $. 
In the time interval $[\TL,\tau^\Lambda]$ the system keeps losing information. 
Differently from the noise in $[0,\TL]$, this noise is \textit{essential} for the following non-Markovian phenomena. Indeed, we have increases $f(t^\Lambda)-f(s)>0$ for all the intervals $[s,t^\Lambda]$ with $s\in (\TL,\tau^\Lambda)$. 
The maximum information backflow is obtained in $[\tau^\Lambda,t^\Lambda]$, when the system recovers information from the environment at all times ($f'(t)>0$). Moreover, at time $t^\Lambda$, the system goes back to the state assumed at time $\TL$ ($f(t^\Lambda)=f(\TL)$), namely when useless noise ended and the essential noise started.

The characteristic function  of the corresponding PNM core $\omL$ is $\overline{f}(t)=f(t+\TL)/f(\TL)$ (see Figure \ref{newfigdepo}). We use Eq.~(\ref{tLoL}) and get the characteristic times $\tau^{\overline \Lambda}\simeq 0.220$ and $t^{\overline \Lambda}\simeq 0.765$ ($T^{\overline\Lambda}=0$ because $\omL$ is PNM).
The total increase of $\overline f (t)$ is $\overline \Delta= \Delta/f(\TL)\simeq 0.491$.
If we compare the non-Markovian effects of $\mL$ and $\omL$, any distinguishability backflow is amplified by a factor $1/f(\TL)\simeq 2.990$ (see Corollary \ref{corpropbackflows}) and through Eqs.~(\ref{ND1}) and (\ref{Nmix1}) we can evaluate the values of the corresponding non-Markovianity measures: $M^D(\omL)\simeq 0.983>M^D(\mL)\simeq 0.328$ and $M^{mix}(\omL)\simeq 0.329>M^{mix}(\mL)\simeq 0.141$.

The main qualitative difference between $\mL$ and the corresponding PNM core $\omL$ is the presence of a time when all the initial information is recovered.  If the system is evolved by $\omL$, any possible type of information is \textit{completely recovered} to its original value at time $t^{\overline\Lambda}$. Indeed, $\overline{f}(t^{\overline\Lambda})=1$ and the dynamical map at this time is equal to the identity, namely 
 $\Lambda_{t^{\overline\Lambda}}=I_S$. For instance, any pair of 
initially orthogonal states $\{\overline\Lambda_t (\rho_{S,1}), \overline\Lambda_t (\rho_{S,2})\}$ goes from being perfectly distinguishable, to non-perfectly distinguishable for any $t\in(0,t^{\overline \Lambda})$ and then back to perfectly distinguishable at time $t^{\overline\Lambda}$. As noticed above, all PNM depolarizing evolutions completely restore the initial information content of the system at time $t^{\overline\Lambda}$, namely $\overline f(t^{\overline\Lambda})=1$ for all PNM depolarizing $\omL$. 
 Finally, we can see how the initial noise in this dynamics is essential for the following non-Markovian phenomena to happen. Indeed, as soon as we take a non-zero time $s\in(0,t^\Lambda)$, we have a distinguishability backflow in the time interval $[s,t^\Lambda]$.

\section{Quasi-eternal non-Markovianity}\label{eternality}

We briefly introduce a qubit model to show the existence of evolutions with $\TL <\tau^\Lambda<t^\Lambda=\infty $ and $\TL =\tau^\Lambda=t^\Lambda$. The example dynamics are taken from the family of quasi-eternal non-Markovian evolutions~\cite{DDSlong}, which generalize the well-known qubit eternal non-Markovian model~\cite{eternal0,eternal,eternal1}. 
First, we define Pauli evolutions as those having dynamical maps with the following form: 
\begin{eqnarray}\nonumber
\Lambda_t (\,\cdot\,) = \sum_{i=0,x,y,z} p_i(t) \sigma_i \, (\,\cdot\,) \, \sigma_i \, ,
\end{eqnarray}
where $\sigma_{x,y,z}$ are the Pauli operators, $\sigma_0=\mathbbm{1}$, and $p_0(t)=1-p_x(t)-p_y(t)-p_z(t)$. The Pauli map is CPTP if and only if $p_{0,x,y,z}(t)\geq 0$.
The easiest way to appreciate the non-Markovian features of Pauli evolutions is given by studying the corresponding master equation, namely the first-order differential equation defining the evolution of the corresponding system density matrix:
\begin{eqnarray}\label{master}
\frac{d}{dt} \rho_S(t)= \sum_{i=x,y,z} \gamma_i(t) (\sigma_i \rho_S(t) \sigma_i -\rho_S(t)) \, , 
\end{eqnarray}
where $\gamma_i(t)$ are time-dependent real functions. 
It can be proven  that $\gamma_i(t)\geq 0$ for all $i=x,y,z$ and $t\geq 0$ if and only if the corresponding evolution $\mL$ is Markovian~\cite{darekk}. 
Moreover, if $\gamma_i(t)+\gamma_j(t)\geq 0$ for all $i\neq j$ and $t\geq 0$, the 
evolution is\hspace{1pt} P-divisible, namely \hspace{1pt}$V_{t,s}$\hspace{1pt} is\hspace{1pt} at  least P (but not
\begin{widetext}

\begin{figure}
\includegraphics[width=0.999\textwidth]{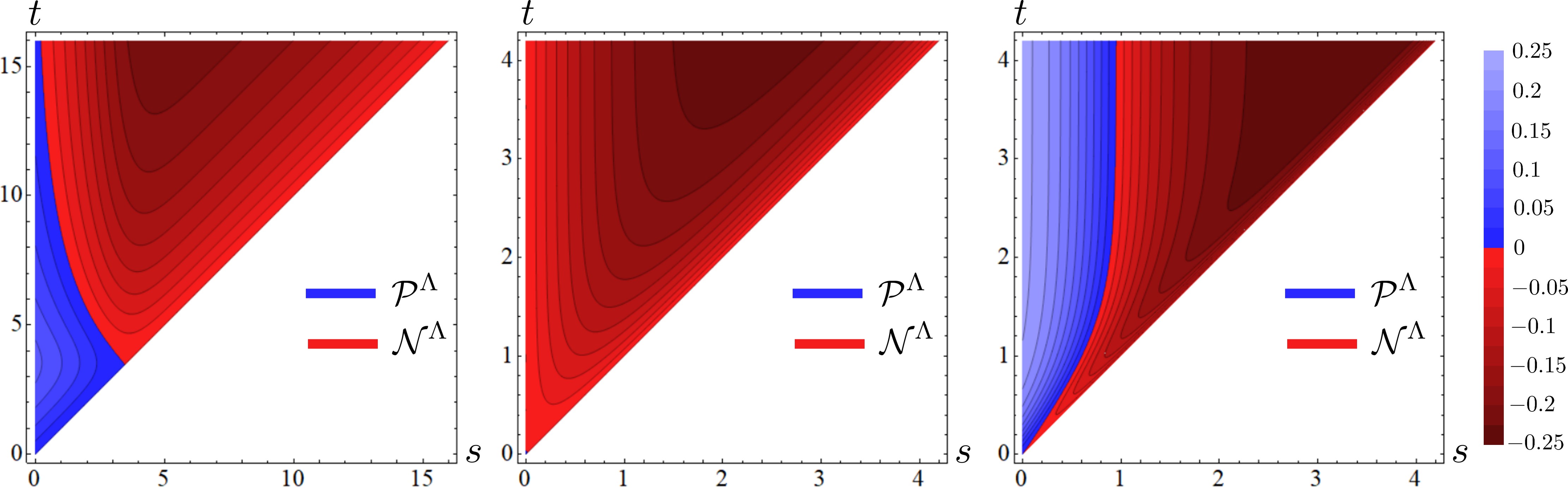}
\caption{$\mathcal P^\Lambda$ and $\mathcal N^\Lambda$ for quasi-eternal PNM evolutions obtained with different $\alpha>0$ and $t_{0}=t_{0,\alpha}$. The pair $\{s,t\}$ belongs to $\mathcal{P}^\Lambda$ ($\mathcal{N}^\Lambda$) when $V_{t,s}$ is CPTP (non-CPTP). 
The lines $\{0,t\}_{t\geq 0}$ and $\{t,t\}_{t\geq 0}$ always belong to $\mathcal P^\Lambda$. The border of $\mathcal N^\Lambda$ contains a pair $\{s,t\}$ along the diagonal $s=t$ if and only if $\gamma_z(t)<0$.
Lighter blue (darker red) colours corresponds to larger (larger in modulo)  positive (negative) values of $\lambda_{t,s}=p_z(s,t)$, the lowest eigenvalue of the Choi operator of $V_{t,s}$.
{\bf Left:} $\alpha=0.1$, $t_{0}=t_{0,\alpha}\simeq 3.465$ and $0=\TL<\tau^\Lambda=t_0<t^\Lambda=\infty$. Infinitesimal intermediate maps $V_{t+\epsilon,t}$ are non-CPTP for all  $t>\tau^\Lambda$, while $V_{t,s}$ is non-CPTP for infinitesimal $s$ and diverging $t$. Indeed, $\TL=0$ and $t^\Lambda=\infty$.  Only property (i) from Corollary \ref{corollaryPNM} is verified.
{\bf Center:}  $\alpha=1$, $t_{0}=t_{0,\alpha}=0$ and $0=\TL=\tau^\Lambda=t^\Lambda$. 
Any intermediate map $V_{t,s}$ with $0<s<t$ is non-CPTP.  Both properties (i) and (ii) from Corollary \ref{corollaryPNM}  are verified.
{\bf Right:}  $\alpha=5$, $t_{0}=t_{0,\alpha}=0$ and  $0=\TL=\tau^\Lambda=t^\Lambda$. We have $\TL=0$ as $V_{t,\epsilon}$ is non-CPTP for infinitesimal $\epsilon$ and $t\in (0,T_\epsilon)$, where $\lim_{\epsilon\rightarrow 0} T_\epsilon \rightarrow 0$. Only property (ii) of Corollary \ref{corollaryPNM} is verified.
 For $s\lesssim 0.9$, there are $s<t_1<t_2$ such that $V_{t_2,s}=V_{t_2,t_1}\circ V_{t_1,s}$ is CPTP, while $V_{t_2,t_1}$ and $V_{t_1,s}$ are not CPTP. 
}
\label{quasieternal}
\end{figure}
\end{widetext}  
 necessarily CP) for all $s\leq t$. 
The probabilities and the rates that define the quasi-eternal model are:
\begin{eqnarray}\nonumber
 &&\!\!\!\!\!\!\!\!\!\!\!\!\!\!\!\!\!\!\! p_{x}(t)=p_y(t)= \frac{1-\expp{-2 \alpha t}}{4}  \, , \\
\label{probs}
&&\!\!\!\!\!\!\!\!\!\!\!\!\!\!\!\!\!\!\!  p_z(t)=\frac{1}{4}	\left( 1 + \expp{-2 \alpha t} - \frac{2 \expp{-\alpha t} \, \mbox{cosh}^{\alpha} (t-t_0)}{\mbox{cosh}^{\alpha} (t_0)}\right) \,  ,
\\ \label{rates}
&&\!\!\!\!\!\!\!\!\!\!\!\!\!\!\!\!\!\!\!  \left\{\gamma_x(t),\gamma_y(t),\gamma_z(t) \right\}= \frac{\alpha}{2} \{1,1,-\mbox{tanh}(t-t_0)\}  ,  
\end{eqnarray}
where these time-dependent parameters generate maps $\Lambda_t$ that are CPTP at all times if and only if $\alpha>0$ and 
$$t_{0}\geq t_{0,\alpha}=\max\{0,(\log(2^{1/\alpha} -1))/2\}, $$
where $t_{0,\alpha}>0$ for $\alpha\in(0,1)$ and $t_{0,\alpha}=0$ for $\alpha\geq 1$~\cite{DDSlong}. 
We 
call quasi-eternal non-Markovian the Pauli  evolutions defined by the probabilities (\ref{probs}), or equivalently the solution of the
master equation (\ref{master}) with rates (\ref{rates}), where $t_0\geq t_{0,\alpha}$.  These evolutions are P-divisible and, since $\gamma_z(t)<0$ for $t>t_{0}$, the infinitesimal intermediate maps are non-CPTP for all $t>t_0$.

The intermediate map of a Pauli evolution assumes the Pauli form $V_{t,s}(\,\cdot\,) = \sum_{i=0,x,y,z} p_i(s,t)  \sigma_i (\,\cdot\,) \sigma_i $, where:
\begin{eqnarray}\nonumber
&&\!\!\!\!\!\!\!\!\!\!p_{x}(s,t)=p_y(s,t)= \frac{1-\expp{-2 \alpha (t-s)}}{4}  \, ,
\\ \nonumber
&&\!\!\!\!\!\!\!\!\!\! p_z(s,t)=\frac{1}{4} + \frac{\expp{-2 \alpha (t-s)}}{4} - \frac{2 \expp{-\alpha (t-s)}  \, \mbox{cosh}^{\alpha} (t-t_0)}{4 \,\mbox{cosh}^{\alpha} (s-t_0)}    \, .
\end{eqnarray}
Notice that, as any Pauli channel, the intermediate map $V_{t,s}$ is CPTP if and only if $p_{0,x,y,z}(s,t)\geq 0$.
The lowest eigenvalue of the Choi state of $V_{t,s}$ is $\lambda_{t,s}=p_z(s,t)$. 
In Figure \ref{quasieternal} we represent $\mathcal P^{\Lambda}$ and $\mathcal N^\Lambda$ for three PNM evolutions from this family, namely the collection of time-pairs $\{s,t\}$ such that $V_{t,s}$ is respectively CPTP and non-CPTP. We see that for $\alpha\in(0,1)$ we have $\TL<\tau^\Lambda=t^\Lambda$, while for $\alpha\geq 1$ we have  $\TL=\tau^\Lambda=t^\Lambda$.

We prove that $\TL=t_0-t_{0,\alpha}$  and $\tau^\Lambda = t_0$ (see Appendix~\ref{TLalpha}).
The latter result is a direct consequence of the form of the master equation, which has negative rates if and only if $t>t_0$. Indeed, $V_{t+\epsilon,t}$ is CPTP for infinitesimal $\epsilon$ if and only if $\gamma_{x,y,z}(t)\geq 0$.

Interestingly, we can appreciate a peculiar scenario for $\alpha>1$, where we obtain a CPTP map through the composition of non-CPTP maps. Without loss of generality, we fix $t_0=t_{0,\alpha}=0$. There exist initial times $s'>0$ such that, $V_{t,s'}$ is non-CPTP for all $t\in(s,t')$, while $V_{t,s'}$ is CPTP for all $t\geq t'$. Notice that, since $\gamma_z(t)< 0$ for all $t>0$, $V_{t+\epsilon,t}$ is non-CPTP for infinitesimal $\epsilon>0$ and all $t> 0$. Therefore, if we consider $t_1<t'<t_2$, we have that $V_{t_1,s}$ is non-CPTP and $V_{t_2,s}$ is CPTP. The latter map can be obtained via the composition of $V_{t_1,s}$ with infinitesimal intermediate maps as follows $V_{t_2,s}=V_{t_2 , t_2 - \epsilon} \circ \dots \circ V_{t_1+\epsilon,t_1} \circ V_{t_1,s}$, namely the CPTP map $V_{t_2,s}$
is obtained by composing infinitesimal non-CPTP maps $V_{t+\epsilon,t}$
with the non-CPTP intermediate map $V_{t_1,s}$. 
The composition of infinitesimal intermediate maps that we wrote corresponds to $V_{t_2,t_1}=V_{t_2 , t_2 - \epsilon} \circ \dots \circ V_{t_1+\epsilon,t_1}$, which, depending on $t_1$ and $t_2$, can be either CPTP or not.

Finally, a simple variation of this model leads to a trivial example of $\TL< \tau^\Lambda = t^\Lambda$, where we exploit condition (C) of Eq.~(\ref{TLambdanuovo}). 
Consider an evolution that is unitary in an initial time interval, namely $\Lambda_t=U_t$ is unitary for $t\in [0,t_U]$, and later it behaves as an eternal PNM evolution with $\alpha>1$ and $t_0=0$. 
Such an evolution would, for instance, be given by integrating Eq. (\ref{master}) with $\{\gamma_x(t),\gamma_y(t),\gamma_z(t) \} = \theta(t-t_U) \{1,1,-\mbox{tanh}(t-t_U)\} $, where $\theta(x)=1$ for $x\geq 0$ and it is zero-valued otherwise. Indeed, in $[0,t_U]$ the evolution would correspond to the identity and $0=\TL<\tau^\Lambda=t^\Lambda=t_U$.

\section{Discussion} 
 We studied the difference between two types of initial noise in non-Markovian evolutions,  where essential noise makes the system lose the same information that takes part during later backflows, while the information lost with useless noise is never recovered. Indeed, this last type of noise can be compared to a Markovian pre-processing of the system.
We identified as PNM those evolutions showing only essential noise, while NNM evolutions have both type of noises. We proved that any NNM evolution can be simulated as a Markovian pre-processing, which generates the useless noise, followed by a PNM evolution, which represents the (pure) non-Markovian core of the evolution.
 In order to distinguish between PNM and NNM, we introduced a temporal framework that aims to describe the timing of fundamental non-Markovian phenomena. We identified the most distinguishable classes arising from this framework, where PNM and NNM evolutions fit naturally. Moreover, several mathematical features connected with this classification have been identified.
 
Later, we focused on the phenomenological side of this topic. Indeed, we addressed the problem of finding which backflows and non-Markovian measures are amplified when PNM evolutions are compared with their corresponding noisy versions, proposing constructive and measurable results within the context of state distinguishability. 
We studied how the entanglement breaking property is lost/preserved when we compare PNM cores and their corresponding NNM evolutions. 
Moreover, we discussed the possibility to activate correlation backflows when we extract the PNM core out of NNM evolutions.
Through several examples we showed how to extract PNM cores, clarified the possible scenarios concerning the timings of non-Markovian phenomena and explained why useless noise has the only role of suppressing the backflows generated by the PNM core.
Finally, it would be interesting to study a further classification that distinguishes between the classical and the quantum content of useless noise, essential noise and, more importantly,  information backflows.

Some dynamical models, such as dephasing and amplitude damping, have PNM cores that go from being non-unitary to unitary, i.e., satisfy Proposition~\ref{propC}. Nonetheless, not all PNM evolutions satisfy this property (see   Section~\ref{secfeatures}). It would be interesting to study which are the minimal conditions under which a given class of evolutions has PNM cores satisfying Proposition~\ref{propC}. A reasonable class could be given by the one-parameter evolutions, as described in Ref.~\cite{DDSlong}, namely those with a single rate in the corresponding Lindblad master equation. More in general, concerning the possibility to lose and completely recover some type of information, it is crucial to understand whether PNM evolutions always enjoy this property, i.e., \textit{``If an evolution is PNM, there exist an initialization that during the dynamics loses and then completely retrieves the information content for at least one quantifier''}.

We generalized the definition of $T^\Lambda$ to non-divisible dynamics. We believe that the extraction of the PNM core of most of the well-known non-divisible models can be obtained within this framework (see Appendix~\ref{noninvdepo} for the study of a non-divisible depolarizing model). Nonetheless, it would be interesting to understand whether and how a PNM core can be extracted when $T^\Lambda =T^{ND}$, with $T^{ND}$ being the non-divisibility starting time. For instance, image non-increasing and kernel non-decreasing NNM evolutions $\mL$ may lead to this exotic scenario \cite{continuity}. If $\mL$ shows non-Markovian effects after a dimensionality reduction of the state space, the PNM core extraction would acquire a more abstract meaning. Indeed, we may require $\omL$ to act on a space with a higher dimensionality if compared to the space of states target of the non-Markovian effects of $\mL$.

We analysed when and to what extent distinguishability backflows are amplified by PNM cores. Moreover, we gave a constructive method to build the states that provide the largest backflows. It would be interesting to understand whether this approach can be generalized to other quantifiers, e.g., distinguishability of state ensembles~\cite{BD}, Fisher information~\cite{Fisher0,Fisher} and correlations~\cite{DDSlong,Janek,DDSDONATO}.
Another interesting topic would be to understand whether PNM evolutions can lead to the activation of other non-Markovian phenomena, as discussed in the context of correlation backflows.

We saw that PNM evolutions have non-Markovianity measures that cannot be smaller than the associated NNM evolutions in Section~\ref{measures}.
Moreover, we gave conditions under which PNM evolutions have strictly larger distinguishability measures. Understanding in which other cases and to what extent this strict inequality can be obtained with other information quantifiers and other non-Markovianity measures is interesting.

\section*{Acknowledgements}

The author would like to thank Antonio Ac{\'i}n and Giulio De Santis for illuminating discussions. 
This work was supported by 
the research project ``Dynamics and Information Research Institute - Quantum Information, Quantum Technologies'' within the agreement between UniCredit Bank and Scuola Normale Superiore di Pisa (CI14\_UNICREDIT\_MARMI),
the Spanish Government (FIS2020-TRANQI and Severo Ochoa
CEX2019-000910-S), the ERC AdG CERQUTE, the AXA
Chair in Quantum Information Science, Fundacio Cellex,
Fundacio Mir-Puig and Generalitat de Catalunya (CERCA,
AGAUR SGR 1381).

\bibliographystyle{quantum}
\bibliography{bib}

\appendix
\widetext

\section{Proof that $\TL$ is given by a maximum and not a supremum}\label{proofclosed}

Given the intermediate map $V_{t,s}$ of a given evolution $\mathbf{\Lambda}$,  we call $\lambda_{t,s}$ the lowest eigenvalue of the corresponding operator given by the Choi-Jamiołkowski isomorphism. Hence, we have that $V_{t,s}$ is CPTP if and only if $\lambda_{t,s}\geq 0$. We call $\Omega_T \subseteq \mathbb{R}^2$ the closed set of pairs of times $\{s,t\}$ such that $s\leq t$ and $s\in [0,T]$. We can rewrite Eq.~(\ref{TLambdanuovo}) as
\begin{equation}\label{TLN}
T^\Lambda=\sup\{\,T \,|\, \mbox{(A+B) } \,\, \lambda_{t,s}\geq 0 \,\,\,\mbox{for all } \{s,t\}\in \Omega_T\},
\end{equation}
where for now we replaced the maximum with the supremum and removed condition (C). Now, we prove that this supremum is indeed a maximum. We call $\mathcal P^\Lambda=\lambda^{-1}_{t,s}([0,\infty))\subseteq \mathbb{R}^2_\leq$ the set of pairs of times such that $\lambda_{t,s}\geq 0$. Being $V_{t,s}$ continuous in $\{s,t\}$, so it is $\lambda_{t,s}$. We notice that, since $\lambda_{t,s}$ is continuous and $\mathcal P^\Lambda\lambda^{-1}_{t,s}([0,\infty))$ is the preimage of the closed set $[0,\infty)$, $\mathcal P^\Lambda=\lambda^{-1}_{t,s}([0,\infty))$ is a closed subset of $\mathbb{R}^2$. Hence, $\mathcal N^\Lambda=\lambda^{-1}_{t,s}((-\infty,0))$ is an open subset of $\{s,t\}_{0\leq s\leq t}$.

Now we prove that Eq.~(\ref{TLN}) is a maximum, namely that $\lambda_{t,s}\geq 0 \,\,\,\mbox{for all } \{s,t\}\in \Omega_{T^\Lambda}$. Suppose that $\lambda_{t,s}$ is not positive semidefinite for all $ \{s,t\}\in \Omega_{T^\Lambda}$. Hence, there must exists a pair of times $\{ \tilde s,\tilde t \}\in \Omega_{T^\Lambda}$ such that   $\lambda_{\tilde t,\tilde s}<0$. Since $\lambda_{t,s}$ is continuous and $\lambda^{-1}_{t,s}((-\infty,0))$ is open, there must exists a small enough $\epsilon>0$ such that $\lambda_{\tilde t,\tilde s-\epsilon}<0$. It follows that $\lambda_{t,s}$ is not positive semidefinite for all $\{s,t\} \in \Omega_{T^\Lambda - \epsilon}$ and therefore we have a contradiction with $T^\Lambda$ being the supremum time $T$ for which $\lambda_{t,s}\geq 0 \mbox{ for all } \{s,t\}\in \Omega_T$. Hence, $\lambda_{t,s}\geq 0$ for all $\{t,s\}\in\Omega_{T^{\Lambda}}$ and Eq.~(\ref{TLN}) is a maximum. Finally, the addition of condition (C) in Eq.~(\ref{TLambdanuovo}) to Eq.~(\ref{TLN}) does not change this result.

\section{Proofs of Lemma~  \ref{ABviolation}, \ref{lemma4} and Proposition~\ref{propDgen}}\label{lemma123}

\begin{proof}[Proof of Lemma~\ref{ABviolation}]
First we prove the first sentence. Since $[0,T]$ is included in $[0,\TL]$, if (A) is satisfied in $[0,\TL]$, then it is also satisfied in $[0,T]$ for $T\leq \TL$. Concerning condition (B) for $T\leq \TL$, we have to distinguish two situations. If $t\in(T,\TL]$, then $V_{t,T}$ is CPTP because the evolution is CP-divisible in $[0,\TL]$. If $t>\TL$, we can write $V_{t,T}=V_{t,\TL}\circ V_{\TL,T}$, which is CPTP as it is the composition of two CPTP maps: $V_{\TL,T}$ is CPTP because the evolution is CP-divisible in $[0,\TL]$ and $V_{t,\TL}$ is CPTP because condition (B) is satisfied for $T=\TL$.
Finally, we obtain the proof of the second sentence by considering that a violation of condition (A) implies that $V_{t,s}$ is not CPTP for some $s<t\leq T$. Hence, (B) is violated at time $s<T$.
\end{proof}

\begin{proof}[Proof of Lemma~\ref{lemma4}]
This result is a direct consequence of Lemma~\ref{ABviolation}. Indeed, for infinitesimal values of $\epsilon>0$, we can say that the condition (B) is violated at time $\TL+\epsilon$. This is the earliest time when this property holds true because, given the result of Lemma~\ref{ABviolation}, if $V_{t,s}$ wouldn't be CPTP for $s\leq \TL$, then Eq.~(\ref{TLambdanuovo}) would not correspond to the maximum time when (A) and (B) are satisfied. 

We prove the last sentence. $\TL\leq t$ otherwise (A) would not be satisfied. Now we show that $\TL\in [s,t)$ leads to a contradiction. We express the non-CPTP intermediate map through the following composition $V_{t,s}=V_{t,\TL}\circ V_{\TL,s}$. The r.h.s. of this equation is CPTP as it is the composition of two CPTP maps: $V_{\TL,s}$ is CPTP for condition (A) and   $V_{t,\TL}$ is CPTP for condition (B). Since the l.h.s. is not CPTP, we have a contradiction: either $V_{\TL,s}$, $V_{t,\TL}$ or both are not CPTP and therefore either (A), (B) or both cannot be satisfied for $T \in [s,t)$.
\end{proof}

\begin{proof}[Proof of Proposition~\ref{propDgen}]
We start with the case $\TL<t^\Lambda$.
We write $V_{t^\Lambda,\TL+\epsilon}=V_{t^\Lambda,s}\circ V_{s,\TL+\epsilon}$, where $V_{t^\Lambda,\TL+\epsilon}$ is not CPTP  and $V_{s,\TL+\epsilon}$ is CPTP because $t^\Lambda$ is the earliest time for which the intermediate map starting from $\TL+\epsilon$ is not CPTP (see Lemma~\ref{lemma4} and Eq.~(\ref{tnm})). Since the composition of two CPTP maps is CPTP, $V_{t^\Lambda,s}$ cannot be CPTP, otherwise $V_{t^\Lambda,\TL+\epsilon}$ would be CPTP. 

%Consider now $\TL=t^\Lambda$. We want to prove that there exists a time  $T>\TL$ that depends on $\mL$ such that $V_{t,\TL+\epsilon}$ is non-CPTP for infinitesimal $\epsilon>0$ and all $t\in (\TL,T)$. If this would not be true, Lemma~\ref{lemma4} would still hold and there would be an infimum time $\theta^\Lambda$ such that $V_{\theta^\Lambda,\TL+\epsilon}$ is non-CPTP for infinitesimal $\epsilon>0$. In other words,  $V_{t,\TL+\epsilon}$ would be CPTP for infinitesimal $\epsilon>0$ and all $t\in [\TL,\theta^\Lambda)$. Hence, this scenario corresponds to $\TL<t^\Lambda=\theta^\Lambda$ and therefore it contradicts the hypothesis $\TL=t^\Lambda$ that we are studying. We conclude that  there exists $t^{*\Lambda}>\TL$ such that  $V_{t,\TL+\epsilon}$ is CPTP for all $t\in (\TL,t^{*\Lambda})$ and infinitesimal $\epsilon>0$.

Consider now $\TL=t^\Lambda$. We want to prove that for all $T>\TL$, the infinitesimal intermediate map  $V_{t+\epsilon,t}$ is non-CPTP for infinitesimal $\epsilon>0$ and infinite times $t$ inside $ (\TL,T)$. Take a $T>\TL$ small enough such that either there are (i) infinite times $t\in(\TL,T)$ when $V_{t+\epsilon,t}$ is non-CPTP or (ii) there are zero. The case of finite times can be avoided by simply considering $T$ (greater but) close enough to $\TL$. We now  analyse case (ii). Since $\TL=t^\Lambda$, the point $\{\TL,\TL\}$ must belong to the border of $\mathcal N^\Lambda$, and therefore there exist a continuum of $\{s,t\}$ infinitesimally close to $\{\TL,\TL\}$ which belong to $\mathcal N^\Lambda$. We remind that $\mathcal N^\Lambda$, the collection of pairs $\{s,t\}$ such that $V_{t,s}$ is non-CPTP, is an open set. Consider a non-CPTP $V_{t_2,t_1}$ such that $\TL<t_1<t_2<T$ and write it as the composition of a large number of infinitesimal intermediate maps. We obtain $V_{t_2,t_1}=V_{t_2,t_2-\epsilon}\circ V_{t_2-\epsilon,t_2-2\epsilon}\circ \dots \circ V_{t_1+\epsilon,t_1}$, where the l.h.s. is non-CPTP while, for small enough $\epsilon>0$, the components on the r.h.s. are all CPTP. The last statement is a consequence of the definition of case (ii). Since the composition of CPTP maps is CPTP, we have a contradiction and only the scenario (i) can take place. 

A pathological case with $0=\TL=t^\Lambda$ such that the times $t$ when $V_{t+\epsilon,t}$ is non-CPTP do not contain a whole interval $(\TL,T)$ where this property is verified is given by the master equation (\ref{master}) with the rates: $\{\gamma_x(t),\gamma_y(t),\gamma_z(t) \} = \{1,1,-\sin(1/t) \mbox{tanh}(t)\}$. 
This is the case because $V_{t+\epsilon,t}$ is non-CPTP at time $t$ if and only if $\gamma_z(t)<0$ and $-\sin(1/t) \mbox{tanh}(t)$ has not got a definite sign in any time interval $(0,T)$.

%If this would not be true, Lemma~\ref{lemma4} would still hold and there would be an infimum time $\theta^\Lambda$ such that $V_{\theta^\Lambda,\TL+\epsilon}$ is non-CPTP for infinitesimal $\epsilon>0$. In other words,  $V_{t,\TL+\epsilon}$ would be CPTP for infinitesimal $\epsilon>0$ and all $t\in [\TL,\theta^\Lambda)$. Hence, this scenario corresponds to $\TL<t^\Lambda=\theta^\Lambda$ and therefore it contradicts the hypothesis $\TL=t^\Lambda$ that we are studying. We conclude that  there exists $t^{*\Lambda}>\TL$ such that  $V_{t,\TL+\epsilon}$ is CPTP for all $t\in (\TL,t^{*\Lambda})$ and infinitesimal $\epsilon>0$.

\end{proof}

\section{Proof that $\TL\leq \tau^\Lambda\leq t^\Lambda$ and more}\label{tlambdavari}

The time $\TL$ cannot be larger than $\tau^\Lambda$ or we would have a violation of condition (A) from Eq.~(\ref{TLambdanuovo}). Since $\tau^{\Lambda}$ is defined through the infimum, $\TL$ and $\tau^\Lambda$ may coincide, but in this case $V_{\tau^{\Lambda}+\epsilon,\tau^{\Lambda}}=V_{\TL+\epsilon,\TL}$ has to be CPTP for all $\epsilon > 0$, otherwise the condition (B) for $\TL$ would be violated. Hence, when $\TL=\tau^\Lambda$, Eq.~(\ref{tnm}) is not a minimum.

Now we prove $\tau^\Lambda\leq t^\Lambda$ by showing that a violation of this inequality leads to a contradiction. If $t^\Lambda< \tau^\Lambda$, we would have that $V_{t^\Lambda,\TL+\epsilon}$  is not CPTP while at the same time $V_{\tau^\Lambda+\epsilon,\tau^\Lambda}$ should be the earliest non-CPTP map for an infinitesimal time interval. These two statements are in contradiction because, if $V_{t^\Lambda,\TL+\epsilon}$ is not CPTP, there must be an infinitesimal time interval $[t_1,t_1+\epsilon]$ contained in $[\TL+\epsilon,t^\Lambda]$ such that $V_{t_1+\epsilon,t_1}$ is not CPTP (see below). Hence, in this case we would have $T^\Lambda+\epsilon\leq t_1\leq t^\Lambda<\tau^\Lambda$. This contradicts Eq.~(\ref{taunm}), which defines $\tau^\Lambda$ as the earliest time $t$ for non-CPTP $V_{t+\epsilon,t}$.

Hence, we only need to prove that if $[s,t]$ is a time interval where $V_{t,s}$ is not CPTP, then there exists an infinitesimal time interval $[t_1,t_1+\epsilon]$ such that $V_{t_1+\epsilon,t_1}$ is not CPTP. For any $\epsilon>0$, we can split $[s,t]$ in subintervals of width $\epsilon$ and consider the composition $V_{t,s}=V_{t,t-\epsilon}\circ V_{t-\epsilon,t-2\epsilon}\circ \dots \circ V_{s+\epsilon,s}$. Since the composition of CPTP maps is CPTP, if $V_{t,s}$ is not CPTP there must be at least one infinitesimal subinterval $[t_1,t_1+\epsilon]$ such that $V_{t_1+\epsilon,t_1}$ is not CPTP.

Now, we prove that $T^\Lambda=\tau^\Lambda$ implies  $\TL=\tau^\Lambda=t^\Lambda$. 
 From the definition of $\tau^\Lambda$ given in Eq.~(\ref{taunm}), it follows that there exists $\overline \delta>0$ such that for all $\delta \in (0,\overline \delta)$, the intermediate map $V_{T^\Lambda+\delta+\epsilon,T^\Lambda+\delta}$ is not CPTP for all $\epsilon \in (0,\epsilon^{(\delta)})$, where $\epsilon^{(\delta)}>0$ depends on $\delta$. On the other hand, $t^{\Lambda,\delta}=\inf\{T\ | V_{T,T^\Lambda+\delta} \mbox{ is not CPTP}\}$. Hence, this infimum is given by $T=T^\Lambda+\delta$, and therefore $t^\Lambda=\TL$.

\section{Proof of Proposition~\ref{propC}}\label{proofpropC}

We start by considering the case where $\mathbf{\Lambda}$ is not characterized by an initial time interval $[0,\delta)$ when the corresponding dynamical maps are all unitary.
$\Lambda_s$ is not unitary for  $s\in (0,\delta)$ and $\Lambda_{t}=U$ is unitary for some ${t}\geq \delta$. We can write the intermediate map starting from an infinitesimal time to $t$ as $V_{{t},\epsilon}=U \circ \Lambda_\epsilon^{-1}$, which is not CPTP for all $\epsilon\in (0,\delta)$. Indeed, since $U$ is unitary, $V_{{t},\epsilon}$ and $\Lambda_\epsilon^{-1}$ have the same eigenvalues and $\Lambda_\epsilon^{-1}$ is non-CPTP because it is the inverse of a non-unitary CPTP map.
It follows that $T^\Lambda=0$. 
Suppose now that $\mathbf{\Lambda}$ is unitary in $[0,\delta)$. Since we assumed that $\Lambda_{s}$ is not unitary for some $s<t$, there exists at least one finite time interval $(\delta, \delta')$ such that  $\Lambda_{s}$ is not unitary for  $s\in(\delta, \delta')$, where  $\delta'\leq t$. Given the conditions (B) and (C) of Eq.~(\ref{TLambdanuovo}), we have to check whether $V_{t,s}$ is CPTP for $s=\delta+\epsilon$ with infinitesimal $\epsilon$. Hence, if we write $V_{t,\delta+\epsilon}=\Lambda_{t}\circ \Lambda^{-1}_{\delta+\epsilon} = U\circ \Lambda^{-1}_{\delta+\epsilon}$ cannot be CPTP because the inverse of a CPTP non-unitary map, namely $\Lambda^{-1}_{\delta+\epsilon}$, is not CPTP and its composition with a unitary transformation $\Lambda_{t}=U$ has the same eigenvalues as $\Lambda^{-1}_{\delta+\epsilon}$. Hence, $V_{t,\delta+\epsilon}$ is not CPTP and $T^\Lambda=0$. 
Finally, we have to prove that there exists at least one intermediate map which is not even positive (P). We call vol$(\Lambda_{t})$ the volume of the image of the evolution at time $t$. Since $\Lambda_{s}$ is CPTP and not unitary, vol$(\Lambda_{s})<$vol$(\Lambda_{0})$~\cite{volume}.
It follows that vol$(\Lambda_{0})$=vol$(\Lambda_{t})<$vol$(\Lambda_{s})$, the intermediate map $V_{t,s}$ cannot be positive~\cite{dividingqm}.

\section{Proof that $M^{mix}(\mL)\leq M^{mix}(\omL)$}\label{appNmix}

Let's say that $M^{mix}(\overline{\mathbf{\Lambda}}) = \overline p$, $\overline{\mathbf{\Lambda}}$ is the PNM core of ${\mathbf{\Lambda}}$  
and 
$
\overline{\mathbf \Lambda}^{mix} = (1-\overline p)\overline{\mathbf{\Lambda}}+\overline p \overline{\mathbf{\Gamma}}^M
$
is Markovian, namely $\overline{\mathbf{\Gamma}}^M$ is optimal to make $\overline{\mathbf{\Lambda}}$ Markovian. The intermediate maps of this evolution are CPTP and read
\begin{eqnarray}\label{incohV}
\overline{V}^{mix}_{t,s }=  \left( (1-\overline p)\overline{{\Lambda}}_t+\overline p \overline{{\Gamma}}^M_t \right) \circ \left( (1-\overline p)\overline{{\Lambda}}_s+\overline p \overline{{\Gamma}}^M_s \right)^{-1} \, .
\end{eqnarray} 
Now we take the Markovian evolution $\mathbf{\Gamma}^M$ and consider $\mathbf{\Lambda}^{mix} = (1-\overline p){\mathbf{\Lambda}}+\overline p {\mathbf{\Gamma}}^M$, where:
\begin{eqnarray}\nonumber 
&{\mathbf{\Gamma}}^M=\left\{
\begin{array}{cc}
\Lambda_t & t< T^{\Lambda} \\
\overline{{\Gamma}}_{t-T^\Lambda} \circ {{\Lambda}}_{T^\Lambda} & t\geq  T^{\Lambda} 
\end{array}  \right.  \, ,& \\ 
\nonumber 
&\mathbf{\Lambda}^{mix} = (1-\overline p){\mathbf{\Lambda}}+\overline p {\mathbf{\Gamma}}^M = 
\left\{
\begin{array}{cc}
\Lambda_t & t< T^{\Lambda} \\
\left( (1-\overline p)\overline{{\Lambda}}_{t-T^\Lambda}  + \overline p \overline{{\Gamma}}_{t-T^\Lambda} \right) \circ \Lambda_{T^\Lambda}  & t\geq  T^{\Lambda} 
\end{array}  \right.   \, . &
\end{eqnarray}
\begin{figure}[H]
\begin{center}
\includegraphics[width=0.73\textwidth]{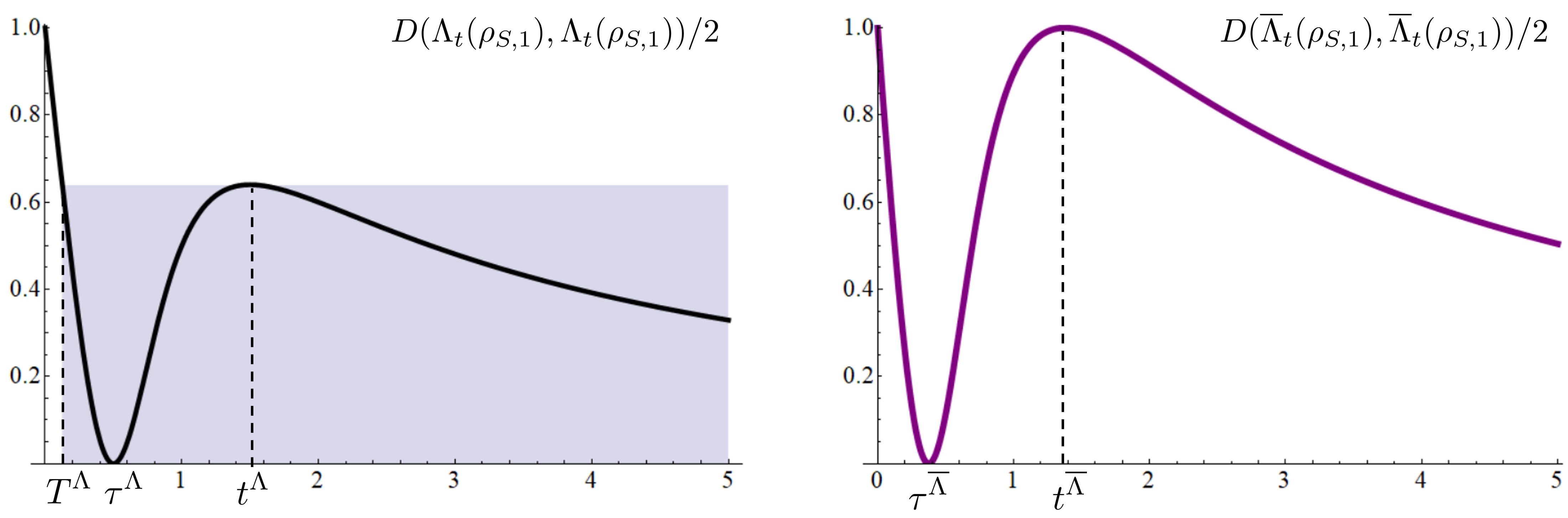} 
\includegraphics[width=0.73\textwidth]{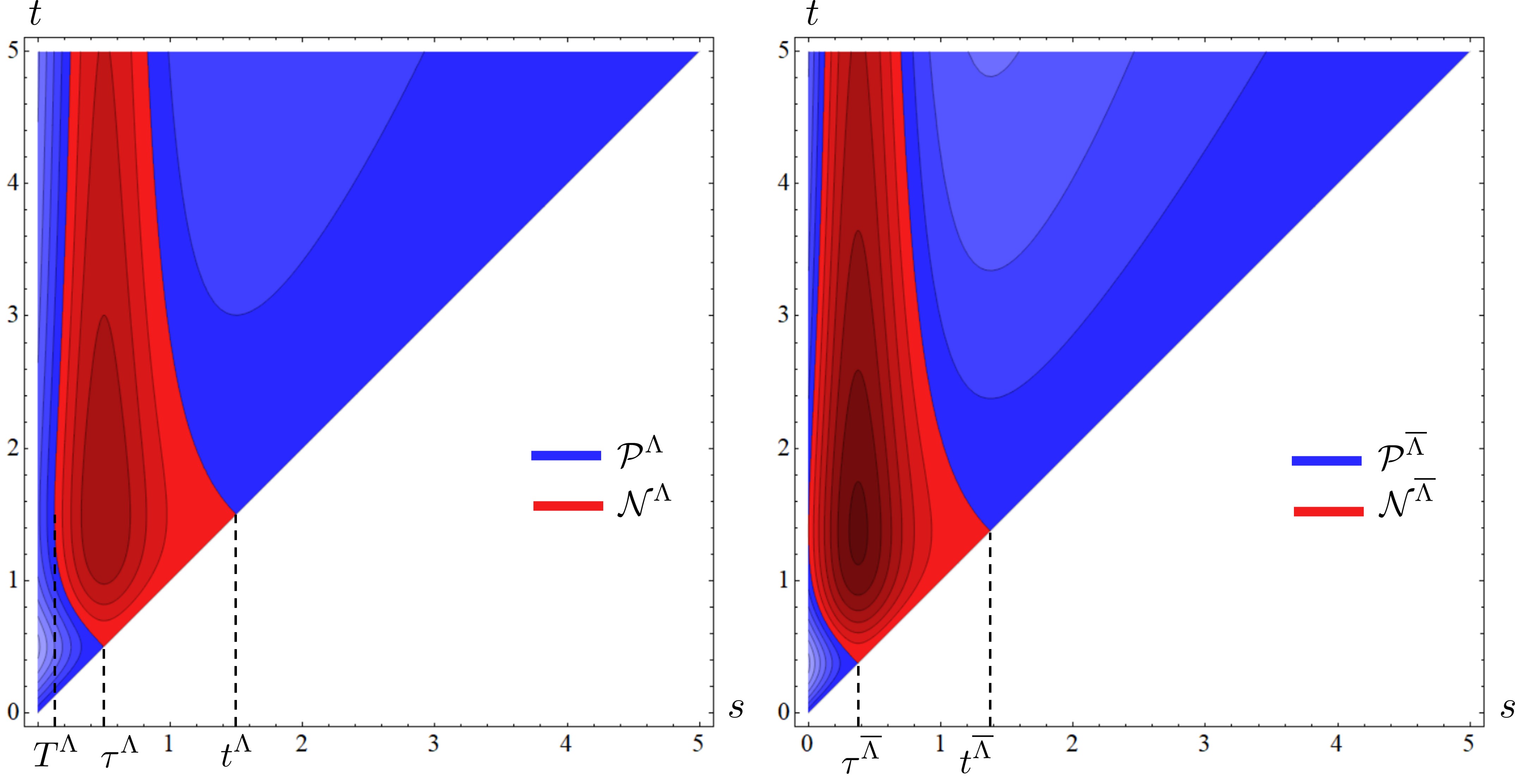}
\end{center}
\caption{
{\bf Top left:}  $D(\Lambda_t(\rho_{S,1}),\Lambda_t(\rho_{S,2}))/2=f(t)$, where $\rho_{S,1,2}$ are orthogonal states and $\mL$ is a NNM depolarizing evolution with $f(t)=(2t-1)^2/(2t^3-t+1)$. 
The NM nature of $f(t)$ is shown after $\TL$: $f(t^\Lambda)-f(s)>0$ for all $s\in(\TL,t^\Lambda)$.
$\tau^\Lambda$ is the first time after which $f(t)$ increases instantaneously.  The total increase of $ f(t)$ is $\Delta=0.64$.
The shaded region represents the PNM core damped by the Markovian pre-processing. All the states are mapped into the maximally mixed state at $\tau^\Lambda=1/2$, when $f(1/2)=0$.
{\bf Top right:} $D(\overline\Lambda_t(\rho_{S,1}),\overline\Lambda_t(\rho_{S,2}))/2=\overline f(t)$, where $\rho_{S,1,2}$ are orthogonal and $\omL$ is a depolarizing model defined by $\overline{f}(t)=f(t+\TL)/f(\TL)$, which is the PNM core of $\mL$. 
All the states are mapped into the maximally mixed state at $\tau^{\overline\Lambda}=1/2$. After this time, the system starts to recover information and at time $t^{\overline\Lambda}=3/2$ all the initial information is restored: $\overline f(t^{\overline\Lambda})=1$ and  $\overline\Lambda_{t^{\overline\Lambda}}=I_S$. The total increase of $\overline f(t)$ is $\overline\Delta=1$.
{\bf Bottom left:} $\mathcal{P}^\Lambda$ and $\mathcal{N}^\Lambda$, the sets of time pairs $\{s,t\}$ when $V_{t,s}$ is respectively CPTP and non-CPTP. Lighter blue (darker red) colours corresponds to larger (larger in modulo)  positive (negative) values of $l_{t,s}=(f(s)-f(t))/4$, which is non-negative if and only if $V_{t,s}$ is CPTP. This evolution is NNM because all the points  $\{s,t\}$ such that $s\leq \TL$ are in $\mathcal P^\Lambda$. Since there exists no CPTP $V_{t,\tau^\Lambda}$ for all $t>\tau^\Lambda$, the time pairs $\{\tau^\Lambda,t\}$ are all inside $\mathcal N^\Lambda$.
{\bf Bottom right:} $\mathcal{P}^{\overline\Lambda}$ and $\mathcal{N}^{\overline\Lambda}$.
Different shades represent different values of $\overline l_{t,s}=(\overline f(s)-\overline f(t))/4 $. Since $\overline V_{t,s}=V_{t+\TL,s+\TL}$, the sets $\mathcal{P}^{\overline\Lambda}$ and $\mathcal{N}^{\overline\Lambda}$ can be obtained by translating $\mathcal{P}^{\Lambda}$ and $\mathcal{N}^{\Lambda}$.
}\label{figdepo}
\end{figure} 
$\!\!\!\!\!\!$Hence, $\mathbf{\Lambda}^{mix} $ is CP-divisible in $t\in[0,T^\Lambda]$ because it corresponds to $\Lambda_t$. Consider the intermediate map $V_{t_2,t_1}^{mix}$ of $\mathbf{\Lambda}^{mix}$, where $T^\Lambda\leq t_1\leq t_2$. We can write it as
\begin{eqnarray}
V^{mix}_{t_2,t_1}\!\!\!&=&\!\!\!{\Lambda_{t_2}}^{mix} \circ ({\Lambda_{t_1}}^{mix})^{-1} =\! \left( \left( (1-\overline p)\overline{{\Lambda}}_{t_2-T^\Lambda}  + \overline p \overline{{\Gamma}}_{t_2-T^\Lambda} \right) \circ \Lambda_{T^\Lambda} \right) \circ \left( \left( (1-\overline p)\overline{{\Lambda}}_{t_1-T^\Lambda}  + \overline p \overline{{\Gamma}}_{t_1-T^\Lambda} \right) \circ \Lambda_{T^\Lambda}  \right)^{-1} = \nonumber \\  \nonumber
\!\!\!&=&\!\!\! \! \left( (1-\overline p)\overline{{\Lambda}}_{t_2-T^\Lambda}  + \overline p \overline{{\Gamma}}_{t_2-T^\Lambda} \right)\! \circ \Lambda_{T^\Lambda}  \circ \Lambda_{T^\Lambda}^{-1}  \circ  \left( (1-\overline p)\overline{{\Lambda}}_{t_1-T^\Lambda}  + \overline p \overline{{\Gamma}}_{t_1-T^\Lambda} \right)^{-1}  = \overline{V}^{mix}_{t_2-T^\Lambda,t_1- T^\Lambda} \, , \end{eqnarray}
which is CPTP because, for $t\geq T^\Lambda$, it is the intermediate map of $\overline{\mathbf{\Lambda}}^{mix}$ (see Eq.~(\ref{incohV})). In summary, the infinitesimal intermediate maps of $\mL^{mix}$ are
\begin{eqnarray} \nonumber
V^{mix}_{t+\epsilon, t}=\left\{
\begin{array}{cc}
V_{t+\epsilon,t} & t< T^{\Lambda} \\
\overline{V}^{mix}_{t-T^\Lambda+\epsilon,t-T^\Lambda} & t\geq  T^{\Lambda} 
\end{array}  \right. \,  ,
\end{eqnarray}
where $V_{t+\epsilon,t}$ are the infinitesimal intermediate maps of $\mathbf{\Lambda}$ which are CPTP for $t\in[0,T^\Lambda]$ and $\overline{V}^{mix}_{t-T^\Lambda+\epsilon,t-T^\Lambda}$, for $t\geq T^\Lambda$,  are the intermediate maps of the Markovian $\overline{\mathbf{\Lambda}}^{mix}$. Hence, $\mathbf{\Lambda}^{mix} = (1-\overline p){\mathbf{\Lambda}}+\overline p {\mathbf{\Gamma}}^M $ is CP-divisible and
\begin{eqnarray} \nonumber
M^{mix}(\mathbf{\Lambda})=\min \{ \,p\,| \,   \exists \mathbf{\Lambda}^M \mbox{ s.t. } (1- p){\mathbf{\Lambda}}+\overline p {\mathbf{\Lambda}}^M \mbox{ is Markovian} \} \leq \overline p = M^{mix}(\overline{\mathbf{\Lambda}})\, .
\end{eqnarray}

\section{Non-invertible depolarizing example}\label{noninvdepo}

We proceed by showing how our framework behaves with a non-bijective depolarizing evolution. 
Consider the characteristic function given by $f(t)=(2t-1)^2/(2t^3-t+1)$ (see Figure \ref{figdepo}). 
This depolarizing evolution is not divisible. As we can see, the evolution is not bijective between $t^{NB}=1/2$ and any later time ($f(1/2)=0$). We cannot define intermediate maps $V_{t,t^{NB}}$ with initial time $t^{NB}$. Nonetheless, Eq.~(\ref{TLambdadepo}) provides a $\TL$ smaller than $t^{NB}$, even without imposing the extra condition $\TL<t^{NB}$. Indeed, $\TL$ is the earliest time such that there exists a time interval $[\TL+\epsilon, t^\Lambda]$ when $V_{t^\Lambda,\TL+\epsilon}$ is not CPTP, while $V_{t^\Lambda,\TL}$ is CPTP (see Proposition~\ref{propDgen}). Hence, $V_{t^\Lambda,\TL+\epsilon}$ not being CPTP implies
$f(t^\Lambda)-f(\TL+\epsilon)>0$, from which we can state that $f(t^\Lambda)>0$. Similarly, from 
 $f(t^\Lambda)-f(\TL)=0$ we have $f(\TL)>0$.

Since evolutions are 
continuous, so are characteristic functions. If $f(0)=1$ and $f(1/2)=0$, there must a be an intermediate time $\TL<1/2$ when the characteristic function assumes the value $f(\TL)>0$ while  $f'(\TL)<0$. This property holds for any possible characteristic function that is zero-valued at one or more times, and therefore finding $\TL$ does not require any additional technique with respect to the divisible case.

Straightforward calculations lead to $\TL=1/8$, $\tau^\Lambda=1/2$ and $t^\Lambda=3/2$. 
Most of the analysis made in Section~\ref{newsecexample} can be done similarly. Nonetheless, we underline some differences coming non-bijectivity. Both $\mL$ and $\omL$ have a time, $\tau^\Lambda=t^{NB}$ and $\tau^{\overline \Lambda}=\overline t^{NB}$ respectively, when {all} the states are mapped into the maximally mixed state, namely when the characteristic function is null.
Nonetheless, only the PNM core $\omL$  \textit{completely retrieves} any possible type of information for at least one time. Indeed,
 $\overline\Lambda_{t^{\overline\Lambda}}=I_S$. For instance, pairs of initially orthogonal states $\rho_{S,1}(t)$, $\rho_{S,2}(t)$ go from perfectly distinguishable ($t=0$), to absolutely indistinguishable ($t=\tau^{\overline\Lambda}$) and back to perfectly distinguishable ($t=t^{\overline\Lambda}$).  

Since we cannot write $V_{t,t^{NB}}$, the corresponding Choi operator lowest eigenvalues $\lambda_{t,t^{NB}}=(1-f(t)/f(t^{NB}))/4$ cannot be evaluated for all $0\leq s \leq t$. Indeed, $f(t^{NB})=0$ and this quantity diverge.  Hence, in Fig. \ref{figdepo} we plot $\mathcal P^\Lambda$ and $\mathcal N^\Lambda$ through the regularization $l_{t,s}=f(s) \lambda_{t,s}=(f(s)-f(t))/4$, which is non-negative if and only if there exists a CPTP $V_{t,s}$ and does not diverge for $s=t^{NB}$. We underline an important subtlety. Even if $V_{t,t^{NB}}$ does not exists, the evolution behaves as non-Markovian in the time intervals $[t^{NB},t]$. Indeed, since Markovianity is defined through CP-divisibility, any other case is labelled as non-Markovian. Undoubtedly, this is the case, where we do not have a specific non-CPTP intermediate map, but since the evolution in $[t^{NB},t]$ is not represented by a CPTP operator, then the evolution must show
 some non-Markovian features. As a matter of fact, the largest non-Markovian features are shown during these time intervals, where states goes from being \textit{identical} to partially, or perfectly (in case of the PNM core and initial orthogonal states), distinguishable. Hence, these effects are not only quantitatively the largest, but they are qualitatively different.

\section{Proof that $T^\Lambda=t_0-t_{0,\alpha}$}\label{TLalpha}

We call $\Lambda^{(\alpha,t_0)}_t$ the dynamical map at time $t$ of the quasi-eternal NM evolution defined by the parameters $\alpha$ and $t_0$. Similarly, we define $V_{t,s}^{(\alpha,t_0)}$.
We need to find the maximum $T$ such that conditions (A) and (B) from Eq.~(\ref{TLambdanuovo}) are satisfied. We start  with condition (B). Since $\Lambda_t^{(\alpha,t_0)}$ is CPTP if and only if $t_0\geq t_{0,\alpha}$ and $V_{t,T}^{(\alpha,t_0)}=\Lambda_{t-T}^{\alpha,t_0-T}$, then $V_{t,T}^{(\alpha,t_0)}$ is CPTP if and only if $T\leq t_0-t_{0,\alpha}$. Instead, we have that an evolution generated from Eq.~(\ref{master}) is CP-divisible in $[0,T]$ if and only if $\gamma_{x,y,z}(t)\geq 0$ for all $t\in [0,T]$. Therefore, condition (A) is satisfied for all $T\leq t_0$. Hence, the maximum $T$ for which conditions (A) and (B)  are simultaneously satisfied is $t_0-t_{0,\alpha}$:  quasi-eternal PNM evolutions have $t_0=t_{0,\alpha}$.

\end{document}